\documentclass[runningheads]{llncs}
\usepackage[T1]{fontenc}
\usepackage{amsmath}
\usepackage{amssymb}
\usepackage{calculi}
\usepackage{microtype}
\usepackage{color} 
\usepackage[bookmarks,bookmarksopen,bookmarksdepth=2]{hyperref}
\usepackage[colorinlistoftodos,prependcaption,textsize=tiny]{todonotes}

\title{Preventing Out-of-Gas Exceptions by Typing\thanks{The work of Luca Aceto, Mohammad Hamdaqa and Stian Lybech was supported by the Icelandic Research Fund Grant No.\@ 218202-05(1-3).}}
\author{%
Luca Aceto\inst{1,2}\orcidID{0000-0002-2197-3018} \and
Daniele Gorla\inst{3}\orcidID{0000-0001-8859-9844} \and
Stian Lybech\inst{1}\orcidID{0000-0001-8219-2285} \and
Mohammad Hamdaqa\inst{1,4}\orcidID{0000-0003-4927-2755}
}

\authorrunning{L. Aceto et al.}
\institute{
Department of Computer Science, Reykjavik University, Iceland \\ \email{\{luca,stian21\}@ru.is} \and
Gran Sasso Science Institute, L'Aquila, Italy \and 
Sapienza, Università di Roma, Rome, Italy \\ \email{gorla@di.uniroma1.it} \and 
Polytechnique Montreal, Montreal, Quebec, Canada \\
\email{mohammad-adnan.hamdaqa@polymtl.ca}
}

\newcommand{\code}[1]{\texttt{#1}}
\newcommand{\TINYSOL}[0]{\textsc{TinySol}}
\newcommand{\EXTRACT}[2][\Delta]{\ensuremath{\mathcal{E}^{\Gamma}_{#1}(#2)}}

\def\VALUES{\ensuremath{\normalfont\text{\sffamily Val}}}
\def\EXPR{\ensuremath{\normalfont\text{\sffamily Exp}}}
\def\STM{\ensuremath{\normalfont\text{\sffamily Stm}}}
\def\LVAL{\ensuremath{\normalfont\text{\sffamily LVal}}}
\def\MVAR{\ensuremath{\normalfont\text{\sffamily MVar}}}
\def\TRANSACTIONS{\ensuremath{\normalfont\text{\sffamily Tr}}}
\def\BLOCKCHAINS{\ensuremath{\SETNAME{B}}}

\def\MNAMES{\ensuremath{\normalfont\text{\sffamily MNames}}} % method names
\def\VNAMES{\ensuremath{\normalfont\text{\sffamily VNames}}} % variable names
\def\FNAMES{\ensuremath{\normalfont\text{\sffamily FNames}}} % field names
\def\ANAMES{\ensuremath{\normalfont\text{\sffamily ANames}}} % address names

\newcommand{\ENV}[1]{\ensuremath{\normalfont\text{\sffamily env}_{#1}}}
\newcommand{\DEC}[1]{\ensuremath{\normalfont\text{\sffamily Dec}_{#1}}}
\newcommand{\SETENV}[1]{\ensuremath{\normalfont\text{\sffamily Env}_{#1}}}

\def\TRUE{\ensuremath{\text{\sffamily T}}}
\def\FALSE{\ensuremath{\text{\sffamily F}}}

\newcommand{\TRANSACT}[5]{\ensuremath{#1\code{->}#2\code{.}#3\code{(}#4\code{):}#5}}
\newcommand{\DEL}[1]{\ensuremath{\normalfont\text{\sffamily del}(#1)}}
\newcommand{\EXC}[1]{\ensuremath{\normalfont\text{\sffamily exc}(#1)}}
\newcommand{\CONF}[1]{\ensuremath{\left<#1\right>}}
\newcommand{\EXTEND}[2]{\ensuremath{\kern-2pt\left[#1 \mapsto #2\right]}}

\def\ITOP{\ensuremath{I^\top}}

\begin{document}
\maketitle
\setcounter{footnote}{0}

\begin{abstract}
We continue the development of \TINYSOL, a minimal object-oriented language based on Solidity, the standard smart-contract language used for the Ethereum platform. 
We first extend \TINYSOL{} with exceptions and a gas mechanism, and equip it with a small-step operational semantics. 
Introducing the gas mechanism is fundamental for modelling real-life smart contracts in \TINYSOL, since this is the way in which termination of Ethereum smart contracts is usually ensured. 
We then devise a type system for smart contracts guaranteeing  that such programs never run out of gas at runtime. 
This is a desirable property for smart contracts, since a transaction that runs out of gas is aborted, but the price paid to run the code is not returned to the invoker. 

\end{abstract}

\section{Introduction}
Smart contracts running on a blockchain can be viewed as a collective system: they are typically deployed in large numbers and interact in an adversarial environment, each pursuing its own individual goal, without specific centralised control. Many languages have been deployed to program them, with
the Ethereum Virtual Machine language \cite{EVM} (EVM) as one of the most successful proposals. It is a Turing-complete language and, thus, it is in principle undecidable whether a given EVM smart contract will halt or not.
This is problematic in a blockchain setting, since smart contracts may be called in transactions, which are executed by miners as part of their attempts to find the next block, as well as when validating existing blocks.
If a transaction were to run forever, it would therefore in effect block the miner from making progress, which for instance could be used to cause a denial of service for the entire network.

To limit such undesirable scenarios, EVM introduced the concept of \emph{gas}, which effectively ensures termination of all transactions.
This mechanism works as follows (see \cite{grishchenko2018} for details):
The gas represents a unit of computation, and thus also a unit of work for the miner.
Each EVM instruction has an associated \emph{gas cost}\footnote{Usually 3--5 units, although some complex operations such as \code{CREATE} have a much higher cost. 
An overview of the EVM instructions and their gas costs can be found at \url{https://evm.codes}.}, and the total gas cost of a transaction is thus proportional to the amount of computational work required by the miner to execute the transaction.
Whenever a user schedules a transaction, two additional parameters have to be specified: the \emph{gas limit} and the \emph{gas price}.
The gas limit is an upper bound on the gas cost of the transaction.
If this limit is exceeded, an out-of-gas exception is thrown, the transaction aborts and its effects are rolled back.
The gas price is the amount of currency that the user is willing to pay to the miner for each unit of gas consumed in the transaction.
This ensures that users will not schedule long-running transactions and just pick an arbitrarily high gas limit, since the gas price will still be paid for each unit of gas consumed in a transaction, even if the transaction aborts with an out-of-gas exception.

The concept of gas is, of course, also inherent in Solidity, since this language compiles to EVM.
However, the original presentation of \TINYSOL{}~\cite{bartoletti2019tinysol}, a minimal object-oriented language based on Solidity, does not include such a mechanism. This limits its ability to realistically model, e.g., aborting transactions and rollback.
In particular, the latter phenomenon is modelled in the operational semantics of that language through a rule with an undecidable premise. 
In our previous work on \TINYSOL{}~\cite{AGL24}, we simply decided to omit this rule, thus in effect allowing non-terminating transactions and also leaving transactions stuck, if an exception-raising %\code{throw} 
command were ever to be reached.

In the present paper, we therefore add a gas mechanism to our version of \TINYSOL{} and present a first attempt at a type system ensuring that transactions never run of out gas at runtime. 
The syntax of the language is provided in Section~\ref{sec:synt} and is essentially the one we used in \cite{AGL24}, where we gave a big-step operational semantics for \TINYSOL{}.
However, handling exceptions in a big-step semantics is rather unnatural, as exceptions can be raised at any intermediate point of a computation. For this reason, a first contribution of this work is to provide a small-step operational semantics for \TINYSOL{}; this is done in Section~\ref{sec:sem}.

We then move to the core contribution of this paper: a type system for guaranteeing that a transaction never runs out-of-gas, that is, an out-of-gas exception is never raised at runtime. 
The type system is defined in Section~\ref{sec:types} and its properties are stated in Section~\ref{sec:prop}. 
Never running out-of-gas is essentially a problem of checking termination of a program, a well-known undecidable problem in its general formulation. 
To make it decidable, and efficiently checkable, we have to reduce the set of programs that the type system accepts; this is discussed thoroughly in Section~\ref{sec:lim}.
We conclude the paper in Section~\ref{sec:concl}, where we also hint at possible directions for future enhancements of this work. 
Due to space limitations, proofs and supplementary material are relegated to the appendices.

\paragraph{Related work}
The literature on methods for proving that programs terminate is vast and we cannot do justice to related work fully in this paper. 
We therefore limit ourselves to mentioning the survey of program termination proof techniques for sequential programs given in~\cite{CookPR11} and the textbook~\cite{NielsonN1992}, which offers an introduction to Hoare-style methods to prove total correctness and to give bounds on the execution time of imperative programs~\cite{Nielson87}. 

In the setting of smart-contract languages, as observed in~\cite{GenetJS20}, proofs of termination are non-trivial even in the presence of a gas mechanism. 
That article presents the first mechanised proof of termination of contracts written in EVM bytecode using minimal assumptions on the gas cost of operations. 
Efficient static-analysis techniques for detecting gas-focused vulnerabilities in smart contracts and their implementation in the tool MadMax have been presented in~\cite{GrechKJBSS18}. 

Among the plethora of studies on type systems for ensuring program termination in several settings, we mention the work by Boudol on typing termination in a higher-order concurrent imperative language; see~\cite{Boudol10} and the references therein. 
Our contribution is inspired by the work by Deng and Sangiorgi in~\cite{PITERM} for the \PI-calculus. 
In particular, we build on their type system 1, which imposes limitations similar to those for our type system, e.g. it prohibits (the \PI-calculus equivalent of) recursive function calls.

\section{\TINYSOL{} with Gas and Exceptions}\label{sec:synt}
\begin{figure}[t]\centering
\begin{syntax}[h]
  DF \in \DEC{F} \IS \epsilon \OR \code{field $p$ := $v$;} DF       \tabularnewline
  DM \in \DEC{M} \IS \epsilon \OR \code{$f(\VEC{x})$ \{ $S$ \} } DM \tabularnewline
  DC \in \DEC{C} \IS \epsilon \OR \code{contract $X$ \{ }           \tabularnewline
                  & & \qquad\ \ \ \code{field balance := $n$; $DF$} \tabularnewline
                  & & \qquad\ \ \ \code{send() \{ skip \} $DM$}     \tabularnewline
                  & & \qquad\ \code{\} $DC$}                        \tabularnewline
  m \in \MVAR \IS \code{this} \OR \code{sender} \OR \code{value}    \tabularnewline
  L \in \LVAL \IS x \OR \code{this.$p$}                             \tabularnewline%
  e \in \EXPR \IS v \OR x \OR m \quad |\quad \code{$e$.balance} \OR \code{ $e$.$p$ } \OR \op(\VEC{e}) \tabularnewline
  S \in \STM  \IS \code{skip} \OR \code{throw} \OR \code{var $x$ := $e$ in $S$} \OR \code{$L$ := $e$} \OR \code{$S_1$;$S_2$}
            \ISOR \code{if $e$ then $S_{\TRUE}$ else $S_{\FALSE}$ } \OR \code{for $e$ do $S$} \OR \code{$e_1$.$f(\VEC{e})$:$e_2$} \tabularnewline
  v \in \VALUES \IS \mathbb{Z} \UNION \mathbb{B} \UNION \ANAMES     \tabularnewline\tabularnewline
\end{syntax}

\noindent where $x, y \in \VNAMES$ (variable names), $p, q \in \FNAMES$ (field names), \\ $X, Y \in \ANAMES$ (address names), $f, g \in \MNAMES$ (method names).
\caption{The syntax of \TINYSOL.}
\label{fig:syntax_tinysol}
\end{figure}

The syntax of \TINYSOL{} is given in Figure~\ref{fig:syntax_tinysol}, where we use the notation $\mathop{\widetilde\cdot}$ to denote (possibly empty) sequences of items.
The set of \emph{values}, ranged over by $v$, is formed by the sets of integers $\mathbb{Z}$, ranged over by $n$, booleans $\mathbb{B} = \SET{\TRUE, \FALSE}$, ranged over by $b$, and address names $\ANAMES$, ranged over by $X, Y$.
 
We introduce explicit declarations for fields $DF$, methods $DM$, and contracts $DC$. 
The latter also encompasses declarations of accounts: an \emph{account} is a contract that contains only the declarations of a special field \code{balance} and of a single special method \code{send()}, which does nothing and is used only for transferring funds to the account. 
By contrast, a contract usually contains other declarations of fields and methods.
For the sake of simplicity, we make no syntactic distinction between an account and a contract but, to distinguish them, we can assume that the set $\ANAMES$ is split into contract addresses and account addresses.

We have four `magic' keywords in our syntax. 
The keyword \code{balance} (type \code{int}) stands for a special field recording the current balance of the contract (or account).
It can be read from, but not directly assigned to, except through method calls.
This ensures that the total amount of currency `on-chain' remains constant during execution. 
The keyword \code{value} (type \code{int}) denotes a special variable that is bound to the currency amount transferred with a method call.
The special variable \code{sender} (type \code{address}) is always bound to the address of the caller of a method. 
Finally, \code{this} (type \code{address}) is a special variable that is always bound to the address of the contract containing the currently executing method.  
The last three of these are local variables, and we collectively refer to them as `magic variables' $m \in \MVAR$.

The declaration of variables and fields ($DV$ and $DF$) are very alike: the main difference is that variable bindings will be created at runtime (and with scoped visibility), hence we can let the initial assignment be an expression $e$; by contrast, the initial assignment to fields must be a value $v$.

The core part of the language is the declaration of expressions $e$ and statements $S$, which are almost the same as in \cite{bartoletti2019tinysol}.
The main differences are: 
(1) we introduce fields $p$ in expressions, instead of keys; 
(2) we explicitly distinguish between (global) fields and (local) variables, where the latter are declared with a scope limited to a statement $S$; 
(3) we introduce explicit \emph{lvalues} $L$, to restrict what can appear on the left-hand side of an assignment (which, in particular, ensures that the special field \code{balance} can never be assigned to directly); and 
(4) we have a \code{for} loop, instead of the more general (and less controllable) \code{while} loop.
This last difference is only introduced for the sake of the type system we are going to present and, indeed, in \cite{AGL24} the \code{while} loop is present.

We now use \TINYSOL{} to describe transactions and blockchains.
A \emph{transaction} is simply a call, where the caller is an account $A$, rather than a contract, together with some gas quantity, used to specify how much computational cost the transaction can consume.
We denote this by writing \mbox{\code{$A$->$X$.$f(\VEC{v})$:$(n,g)$}}, which expresses that the account $A$ calls the method $f$ on the contract (residing at address) $X$, with actual parameters $\VEC{v}$, transferring the amount of currency $n$   with the call, and with an upper bound on the computational cost of the call given by the gas amount $g$ (a positive integer).
Hence, blockchains can be defined as follows:

\begin{definition}[Syntax of blockchains]
A \emph{blockchain} $B \in \BLOCKCHAINS$ is a list of initial contract declarations $DC$, followed by a sequence of transactions $T \in \TRANSACTIONS$:
\begin{center}
\begin{syntax}[h]
  B \in \BLOCKCHAINS  \DCLSYM \code{$DC$ $T$} \qquad\qquad
  T \in \TRANSACTIONS \DCLSYM \epsilon \OR \TRANSACT{A}{X}{f}{\VEC{v}}{(n,g)},\, T 
\end{syntax}
\end{center}
%
%\noindent 
where $g>0$ is the \emph{gas limit} of the transaction.
Notationally, a blockchain with an empty $DC$ will be simply written as the sequence of transactions.
\end{definition}

\section{A Small-step Semantics with Exceptions and Gas}\label{sec:sem}
The semantics given in \cite{AGL24} describes the execution of a complex statement as one single computational step.
This is sensible from the point of view of a user, since a transaction is a method call 
%(the rule for transactions invokes the rule for method calls %in the premise), 
and transactions are atomic.
However, we can also create a small-step semantics for statements in \TINYSOL{} to describe the individual computational steps involved; this form of semantics is more suitable than a big-step one when exception handling is introduced.

To define the semantics, we need some environments to record the bindings of variables (including the three magic variable names \code{this}, \code{sender} and \code{value}), fields, methods, and contracts.
We define them as sets of partial functions:

\begin{definition}[Binding model]
The binding model consists of the following:
\begin{align*}
  \ENV{V} \in \SETENV{V} & : \VNAMES \UNION \MVAR \PARTIAL \VALUES                &
  \ENV{S} \in \SETENV{S} & : \ANAMES \PARTIAL \SETENV{F}                          \\
  \ENV{F} \in \SETENV{F} & : \FNAMES \UNION \SET{\code{balance}} \PARTIAL \VALUES &
  \ENV{T} \in \SETENV{T} & : \ANAMES \PARTIAL \SETENV{M}                          \\
  \ENV{M} \in \SETENV{M} & : \MNAMES \PARTIAL \VNAMES^* \times \STM   
\end{align*}
We write each environment $\ENV{X}$, for  $X \in \SET{V, F, M, S, T}$, as an unordered sequence of pairs $(d, c)$ where $d \in \DOM{\ENV{X}}$ and $c \in \CODOM{\ENV{X}}$.
The notation $\ENV{X}\EXTEND{d}{c}$ denotes the update of $\ENV{X}$ mapping $d$ to $c$.
We write $\ENV{X}^{\EMPTYSET}$ for the empty $X$-environment.
\end{definition}

To simplify the notation, when two or more environments appear together, we shall use the convention of writing the subscripts together; e.g., we write $\ENV{MF}$ instead of $\ENV{M}, \ENV{F}$, and $\ENV{SV}'$ instead of $\ENV{S}', \ENV{V}'$.

Statements are executed in the context of the following environments: a \emph{method table} $\ENV{T}$, which maps addresses to method environments, and a \emph{state} $\ENV{SV}$, which maps addresses to lists of fields and their values, and local variables to their values.
Thus, for each contract, we have the list of methods it declares and its current state.
Of course, the method table is constant, once all declarations are performed, whereas the state will change during the evaluation of a program.

The semantics relies on the concept of a \emph{stack} $Q$ that holds the code that has not yet been executed and the variable environments needed to switch execution frames during a call. 

\begin{definition}[Syntax of stacks]\label{def:stack_syntax}
The syntax of a \emph{stack} $Q$ is given as follows:
\begin{center}\normalfont
\begin{syntax}[h]
  Q \IS \bot \OR S :: Q \OR \ENV{V} :: Q \OR \DEL{x} :: Q \OR \EXC{l} :: Q \tabularnewline
  l \IS \code{rte} \OR \code{neg} \OR \code{oog} \OR \code{pge}
\end{syntax}
\end{center}
\end{definition}

In the above definition, $\bot$ denotes the bottom of the stack, $S$ is a statement, $\ENV{V}$ is a variable environment, $\DEL{x}$ marks the end of the scope of a locally declared variable $x$, and $\EXC{l}$ denotes raising of an exception of kind $l$, whose possibilities are as follows:
$\code{rte}$ denotes an exception raised at runtime because of ill-formed code (e.g., when attempting to access a field not present in a contract or when invoking a method not provided by a contract);
$\code{neg}$ denotes an exception raised when a contract's balance becomes negative; 
$\code{oog}$ denotes an exception raised when a transaction consumes all its gas; and
$\code{pge}$ denotes an exception raised by a $\code{throw}$ command. 

In the small-step semantics, we always execute the element at the top of the stack.
Transitions have the form $\ENV{T} \vdash \CONF{Q, \ENV{SV}, g} \trans \CONF{Q', \ENV{SV}', g'}$, where $g\geq g' \geq 0$, and are defined by the rules in Figure~\ref{fig:tinysolgas_semantics_statements1_sss}, which rely on the (standard) semantics for expression evaluation (written $\trans_e$) taken from \cite{AGL24} and whose rules are in Appendix~\ref{app:omitted_semantics_rules}.
To keep the language and the semantics easy, we make a few simplifying assumptions: 
\begin{enumerate}
\item We assume a unitary gas cost for all \TINYSOL{} commands, rather than having different prices for different instructions.
\item We only let commands consume gas;
    this is in contrast to EVM, which also has gas costs for expression operations such as \code{ADD}, \code{MUL} etc.
\item We omit the \emph{gas price}, since the \TINYSOL{} model of blockchains does not include the concept of miners and blocks.
    Instead, we shall simply say that the actual amount of gas consumed in a transaction is subtracted from the user's account \code{balance}.
\end{enumerate}

\begin{figure}\centering
\begin{semantics}
  \RULE[ss-skip][ts2_sss_skip](g \geq 1) 
    { }
    { \ENV{T} \vdash \CONF{\code{skip} :: Q, \ENV{SV}, g} \trans \CONF{Q, \ENV{SV}, g-1} }

  \RULE[ss-seq][ts2_sss_seq](g \geq 1)
    { }
    { \ENV{T} \vdash \CONF{\code{$S_1$;$S_2$} :: Q, \ENV{SV}, g} \trans \CONF{S_1 :: S_2 :: Q, \ENV{SV}, g} }

  \RULE[ss-if][ts2_sss_if](g \geq 1)
    { \ENV{SV} \vdash e \trans_e b \in \SET{\TRUE, \FALSE}}
    {
    \begin{array}{r @{~} l}
      \ENV{T} \vdash  & \CONF{\code{if $e$ then $S_{\TRUE}$ else $S_{\FALSE}$} :: Q, \ENV{SV}, g} \\
      & \qquad \trans \CONF{S_b :: Q, \ENV{SV}, g-1} 
    \end{array}
    }

  \RULE[ss-for$_\TRUE$][ts2_sss_fortrue](g \geq 1)
    { \ENV{SV} \vdash e \trans v \AND v \geq 1 \AND v' = v-1}
    { 
    \begin{array}{r @{~} l}
      \ENV{T} \vdash & \CONF{\code{for $e$ do $S$} :: Q, \ENV{SV}, g}        \\ 
      & \qquad \trans \CONF{S :: \code{for $v'$ do $S$} :: Q, \ENV{SV}, g-1} 
    \end{array}
    } 

  \RULE[ss-for$_\FALSE$][ts2_sss_forfalse](g \geq 1)
    { \ENV{SV} \vdash e \trans v \AND v < 1}
    { \ENV{T}  \vdash \CONF{\code{for $e$ do $S$} :: Q, \ENV{SV}, g} \trans \CONF{Q, \ENV{SV}, g-1} } 

  \RULE[ss-decv][ts2_sss_decv](g \geq 1)
    { x \notin \DOM{\ENV{V}} \AND \ENV{SV} \vdash e \trans_e v }
    { 
    \begin{array}{r @{~} l}
      \ENV{T} \vdash & \CONF{\code{var $x$ := $e$ in $S$} :: Q, \ENV{SV}, g}   \\
      & \qquad \trans \CONF{S :: \DEL{x} :: Q, \ENV{S}, ((x,v) , \ENV{V}), g-1} 
    \end{array}
    }

  \RULE[ss-assv][ts2_sss_assv](g \geq 1)
    {x \in \DOM{\ENV{V}} \AND \ENV{SV} \vdash e \trans_e v }
    { 
    \begin{array}{r @{~} l}
      \ENV{T} \vdash & \CONF{\code{$x$ := $e$} :: Q, \ENV{SV}, g} \\
      & \qquad \trans \CONF{Q, \ENV{S}, \ENV{V}\EXTEND{x}{v}, g-1} 
    \end{array}
    }

  \RULE[ss-assf][ts2_sss_assf](g \geq 1)
    {
    \begin{array}{l}
      \ENV{V}(\code{this}) = X \AND \ENV{S}(X) = \ENV{P} \\ 
      p \in \DOM{\ENV{P}} \AND \ENV{SV} \vdash e \trans_e v 
    \end{array}
    }
    { 
    \begin{array}{r @{~} l}
      \ENV{T} \vdash & \CONF{\code{this.$p$ := $e$} :: Q, \ENV{SV}, g} \\
      & \qquad \trans \CONF{Q, \ENV{S}\EXTEND{X}{\ENV{P}\EXTEND{p}{v}}, \ENV{V}, g-1}  
    \end{array}
    }

  \RULE[ss-call][ts2_sss_call](g \geq 1)
    { 
    \begin{array}{lllll}
      \ENV{SV} \vdash e_1 \trans_e Y
      &\quad &
      \ENV{SV} \vdash e_2 \trans_e n
      &\quad &
      \ENV{SV} \vdash \VEC{e} \trans_e \VEC{v}
    \\
      \ENV{V}(\code{this}) = X             
      &\quad &
      \ENV{S}(X) =\ENV{P}^X
      &\quad &
      n \leq \ENV{P}^X(\code{balance})
      \\
      \ENV{T}(Y))(f) = (\VEC{x}, S)
      &&
      |\VEC{x}| = |\VEC{v}| = k
      &&
      \ENV{S}(Y) = \ENV{P}^Y
      \\
    \multicolumn{5}{l}{\ENV{V}' \!=\! (\code{this}, Y), (\code{sender}, X) , (\code{value}, n) , (x_1, v_1) , \ldots , (x_k, v_k)  }
    \\
    \multicolumn{5}{l}{\ENV{S}' = \ENV{S}
    \EXTEND{X}{\ENV{P}^X [\code{balance \!-=\! } n]}
    \EXTEND{Y}{\ENV{P}^Y [\code{balance \!+=\! } n]} }
    \end{array}
    }
    { \ENV{T} \vdash \CONF{\code{$e_1$.$f$($\VEC{e}$):$e_2$} :: Q, \ENV{SV}, g} \trans \CONF{S :: \ENV{V} :: Q, \ENV{SV}', g-1} }

  \RULE[ss-throw][ts2_sss_throw](g \geq 1)
    { }
    { \ENV{T} \vdash \CONF{\code{throw} :: Q, \ENV{SV}, g} \trans \CONF{\EXC{\code{pge}} :: Q, \ENV{SV}, g} }

  \RULE[ss-oog][ts2_sss_oog]
    { }
    { \ENV{T} \vdash \CONF{S :: Q, \ENV{SV}, 0} \trans \CONF{\EXC{\code{oog}} :: S :: Q, \ENV{SV}, 0} }

  \RULE[ss-delv][ts2_sss_delv]
    { }
    { \ENV{T} \vdash \CONF{\DEL{x} :: Q, \ENV{S}, ((x,v) , \ENV{V}), g} \trans \CONF{Q, \ENV{SV}, g} }

  \RULE[ss-return][ts2_sss_return]
    { }
    { \ENV{T} \vdash \CONF{\ENV{V}' :: Q, \ENV{SV}, g} \trans \CONF{Q, \ENV{S}, \ENV{V}', g} }

\end{semantics}
\caption{Small-step semantics transition rules for statements with gas.}
\label{fig:tinysolgas_semantics_statements1_sss}
\end{figure}

With these assumptions, gas consumption is modelled by introducing the side condition $g \geq 1$ in most of the rules, since $g$ is decremented by $1$ in the reduct of these rules.\footnote{%
It is conceptually easy, albeit notationally cumbersome, to provide each operation and command with its own cost $c(\cdot)$ and modify the semantics to take such costs into account in the side conditions and conclusions of each rule. 
This would relax the simplifying assumptions 1 and 2 given above.
}
The only rules that do not consume gas are \nameref{ts2_sss_seq}, since this rule does not describe the execution of an operation but merely transforms the top element of the stack, and \nameref{ts2_sss_throw}. 
Notice, however, that those two rules  have the side condition $g \geq 1$ to ensure determinism of the operational semantics: in this way, when the gas becomes 0, only rule \nameref{ts2_sss_oog} can be used.
Finally, the only rules that do not check the remaining amount of gas are \nameref{ts2_sss_delv} and \nameref{ts2_sss_return}, since these rules do not correspond to commands, but are only used to restore a previous state of the variable environment.

Rule \nameref{ts2_sss_throw} says that a \code{throw} command is simply converted into an exception $\EXC{\code{pge}}$, denoting that the exception was thrown from within the program.
The remaining exception labels are used if one of the side conditions fails; so, the execution of most of the statements can lead to an exception in several different ways:
\begin{itemize}
\item An \textbf{o}ut-\textbf{o}f-\textbf{g}as exception \code{oog} is thrown if we have a command on top of the stack but no gas is left, see rule \nameref{ts2_sss_oog}.
\item A \textbf{neg}ative balance exception \code{neg} is thrown if executing the statement would lead to a negative \code{balance}.
  Hence, this exception can only be thrown from the rule \nameref{ts2_sss_call}, if the side condition $n \leq \ENV{P}^X(\code{balance})$ does not hold.
\item Finally, a \textbf{r}un\textbf{t}ime \textbf{e}xception \code{rte} is thrown if any of the other side conditions does not hold. 
  We omit the rules for generating these exceptions, since they are heavy and do not add anything to our presentation. 
  Furthermore, by using a very standard type system, such exceptions can be easily prevented at compile time.
\end{itemize}

Notice that, when an exception appears on top of the stack, the execution halts as there are no rules to conclude a transition in such a case. 
Furthermore, the exception label $l$ has no meaning in the semantics, and the different labels are used only to distinguish between different situations that raise an exception. 
We could have added a command $\code{catch}$ to handle exceptions according to their label (like in many programming languages). 
However, since the purpose of this paper is to develop techniques that prevent out-of-gas exceptions through static checks, for the sake of simplicity we preferred to omit such a command.

%To avoid repeating the rules for each statement two or three times, we shall therefore simply assume rules of the form
%\begin{center}
%\begin{semantics}
%  \RULE[SS-<Stmt>-E]
%    { }
%    { \ENV{T} \vdash \CONF{S :: Q, \ENV{SV}, g} \trans \CONF{ \EXC{l} :: S :: Q, \ENV{SV}, g} }
%\end{semantics}
%\end{center}
%
%\noindent for each of the rules in Figure~\ref{fig:tinysolgas_semantics_statements1_sss} and Figure~\ref{fig:tinysolgas_semantics_statements2_sss} that has a side condition, and with $l$ set according to the above list.

\begin{figure}[t]
\begin{minipage}{0.5\textwidth}
\center
\begin{semantics}
  \RULE[Gen][ts2_trans_genesis]
    { \CONF{DC, \ENV{ST}^{\EMPTYSET}} \trans_{DC} \ENV{ST} }
    { \CONF{\code{$DC$ $T$}, \ENV{ST}^{\EMPTYSET}} \trans_B \CONF{T, \ENV{ST}} }
\end{semantics}
\end{minipage}
\begin{minipage}{0.5\textwidth}
\center
\begin{semantics}
  \RULE[Rev][ts2_trans_revel]
    { }
    { \CONF{\epsilon, \ENV{ST}} \trans_B \ENV{ST} }
\end{semantics}
\end{minipage}

\vspace*{.5cm}

\begin{semantics}
  \RULE[Tx$_1$][ts2_trans_trans1]
    { 
    \begin{array}{l}
  \ENV{S}(A) = \ENV{P}^A
  \qquad
  g \leq \ENV{P}^A(\code{balance}) - n
    \\
      \ENV{T} \vdash \CONF{\code{$X$.$f$($\VEC{v}$):$n$} :: \bot, \ENV{S}, (\code{this}, A), g } \trans\CONF{\bot, \ENV{S}', \ENV{V}, g'} 
       \\
    \ENV{S}'' = \ENV{S}'\EXTEND{A}{ \ENV{P}^A [\code{balance -= } (g - g')] }       
    \end{array}
    }
    { \CONF{\TRANSACT{A}{X}{f}{\VEC{v}}{(n,g)}, T, \ENV{ST}} \trans_B \CONF{T, \ENV{S}'', \ENV{T}} }

  \RULE[Tx$_2$][ts2_trans_trans2]
    { 
    \begin{array}{l} 
  \ENV{S}(A) = \ENV{P}^A
  \qquad
  g \leq \ENV{P}^A(\code{balance}) - n
    \\
      \ENV{T} \vdash \CONF{\code{$X$.$f$($\VEC{v}$):$n$} :: \bot, \ENV{S}, (\code{this}, A), g } \trans \CONF{\EXC{l} :: Q, \ENV{S}', \ENV{V}, g'} 
      \\
      \ENV{S}'' = \ENV{S}\EXTEND{A}{ \ENV{P}^A [\code{balance -= } (g - g')] } 
    \end{array}
    }
    { \CONF{\TRANSACT{A}{X}{f}{\VEC{v}}{(n,g)}, T, \ENV{ST}} \trans_B \CONF{T, \ENV{S}'', \ENV{T}} }

\end{semantics}
\caption{Transition rules for blockchains.}
\vspace*{-.4cm}
\label{fig:tinysolgas_semantics_blockchains}
\end{figure}

The semantics for blockchains is given in Figure~\ref{fig:tinysolgas_semantics_blockchains}
and contains rules \nameref{ts2_trans_genesis} and \nameref{ts2_trans_revel}, to activate and reveal the outcome of a sequence of transactions. Notice that the first rule exploits the semantics for declarations (written $\trans_{DC}$) taken from \cite{AGL24} for turning a sequence of declarations $DC$ into a proper environment; its rules are in Appendix~\ref{app:omitted_semantics_rules}.
The key rules for transactions are \nameref{ts2_trans_trans1} and \nameref{ts2_trans_trans2} for describing the execution of a transaction that either succeeds or fails.
Rule \nameref{ts2_trans_trans1} describes the situation where the transaction succeeds.
Having a gas limit $g$, the side condition requires that $g$ must not be greater than the available \code{balance} of the account $A$ initiating the transaction, \emph{after} any funds $n$ transferred along with the call have been deducted from the \code{balance} (since that could otherwise allow the \code{balance} of $A$ to become negative).
After the transaction succeeds, some amount of gas will have been consumed.
The remaining gas $g'$ is then used to calculate how much gas was used, and then this too is subtracted from the \code{balance} of $A$.
Rule \nameref{ts2_trans_trans2} describes the situation where the transaction fails with an exception $\EXC{l}$.
Here, the gas difference $g - g'$ is \emph{still} subtracted from the \code{balance} of $A$, but otherwise the state $\ENV{S}$ is unaltered.
Thus, all other effects of the transaction that raised the exception are rolled back, and the execution can continue with the next transaction.

% TODO: What if the side conditions of a _transaction_ fail? I.e. in rules [Tx1] or [Tx2]? 
% They can't all be known statically, e.g. the balance, so they will have to be checked at runtime.
% Do we just skip the transaction then?

\section{A Type System for Termination}
The gas mechanism ensures that all transactions will eventually terminate, there\-by preventing denial-of-service attacks resulting from diverging method calls.
However, it does not prevent potentially diverging contracts from being added to the blockchain.
This presents a problem to the users of the smart contracts: how to ensure that they do not inadvertently create a transaction invoking a diverging contract, since that would consume all available gas supplied to the transaction.
Furthermore, even if a method call does not actually diverge, it may still require many execution steps, and the user must therefore be sure to supply sufficient gas for the transaction to allow the execution to finish.
In the present section, we shall describe a type system for statically checking termination of transactions, whilst also providing upper (and lower) bounds on the number of execution steps.

\subsection{The Type System}\label{sec:types}
We wish to create a type system that will allow us to derive an upper bound on the number of execution steps required to run a statement $S$, since that is proportional to the required amount of gas.
Hence, type judgments for statements will be of the form $\Gamma \vdash S : n$ where $n$ is an integer denoting the maximum number of steps.
This number may be affected by the (integer) values of the loop condition expressions in $S$; thus we need the integer type (for both variables and expressions) to provide explicit upper and lower bounds on those values.
We write this as $\TINT^u_\ell$, where $u$ (resp.\@ $\ell$) is the upper (resp.\@ lower) bound (both included).

The $\TINT^u_\ell$ type is quite limiting; for instance, it will disallow some common operations such as increments/decrements of a variable (e.g., \code{$x$ := $x$+1}).
Therefore, we shall also include an ordinary (unbounded) integer type $\TINT$ for integer-valued variables and fields that do not in any way contribute to the expressions that guard \code{for} loops. 
To allow variables of both types to appear in operations, we shall also introduce a subtyping relation $\SUBS$ to allow bounded integer types to be coerced up to the unbounded type (see the rules in Appendix~\ref{app:subtyping}).

We also need the type of methods to have an upper bound on the number of steps required to execute the method body.
This type is written $\TPROC{\VEC{B}}^u_\ell : n$, to be reads as: assuming that the formal parameters are of (base) types $\VEC{B}$ and that the transferred amount of currency is in the interval $[\ell .. u]$, then the method body requires at most $n$ steps to execute.
The bounds on the \code{value} variable are necessary, since the variable might be used in a loop expression\footnote{%
We could also introduce an `unbounded' method type (similarly to what we did for $\TINT$ and $\TINT^u_\ell$), but we shall omit this to avoid complicating the matter further.
The `bounded' method type also has the added benefit of allowing a programmer to specify a minimum amount of currency that must be provided to call a method.}.  
Thus, we also handle divergence resulting from recursive method calls, since any calls within the method body would need to have a value strictly lower than $n$ (as calling the method itself also requires one step).

Finally, we need to give types to contract names, which must provide information on the signatures of fields and methods implemented in the contract.
Here, we shall employ a method similar to the one used in \cite{AGL24} and assume a set of \emph{type names} (or `interface names') $\TNAMES$, ranged over by $I$.
We then let type environments $\Gamma$ be partial functions from names to types, but also include $\Gamma$ in the language of types and use these as the types of interface names.
The idea is that, if $X$ is a contract name of type $I$, then $\Gamma(X) = I$ and $\Gamma(I) = \Gamma_I$, where $\Gamma_I$ then contains the signatures of the fields and methods of contract $X$.
We shall assume that these interfaces are declared besides the contract declarations, e.g.\@ using an interface definition language similar to the one described in \cite{AGL24}, and we refer the reader to that work for further details.

\begin{definition}[Language of types]
We use the following language of types, where $I \in \TNAMES$ is a \emph{type name} (or `interface name') and $\NAMES = \ANAMES \UNION \FNAMES \UNION \VNAMES \UNION \MNAMES \UNION \TNAMES$:
\begin{center}\normalfont
\begin{syntax}[h]
  B                \IS \TBOOL \OR \TINT \OR \TINT^u_\ell \OR I     \\
  T \in \TYPES     \IS B \OR \TPROC{\VEC{B}}^u_\ell : n \OR \Gamma \\
  \Gamma, \Delta   \IS \NAMES \PARTIAL \TYPES                  
\end{syntax}
\end{center}

\noindent We write $\VEC{T}$ for a tuple of types $(T_1, \ldots, T_n)$.
\end{definition}

As in \cite{AGL24}, we require that all interface declarations be \emph{well-formed} in the sense that they must at least contain a declaration for the mandatory members, i.e.\@ the \code{balance} field (type $\TINT$) and the \code{send()} method (type $\TPROC{}^{\code{INT\_MAX}}_0 : 1$).
This ensures that we can define a minimal interface declaration called $\ITOP$, containing just the signatures of \code{balance} and \code{send()}, such that every well-formed interface declaration is a specialisation of $\ITOP$; this is ensured by the subtyping relation and it is necessary to allow us to give a type to the `magic' variable \code{sender}, which is available within the body of any method.
% Again, the definition is similar to the one given in \cite{AGL24}; since the issue of subtyping is orthogonal to the purpose of the present type system, this definition is relegated to Appendix~\ref{app:subtyping}.

We shall also use a \emph{typed} syntax of \TINYSOL{}, where local variables are now declared as $\code{$B$ $x$ := $e$}$ and $B$ is the type of the value of the expression $e$.
Furthermore, we assume that the type information is also stored in $\ENV{V}$, along with the actual value of each variable.
We write this as a triple $(x,v,B)$ but otherwise ignore the type information in the semantics.
Interface types could also be explicitly given in the contract definitions, i.e.\@ as \code{contract $X$ : $I$ \{ $DF$ $DM$ \}}, but to simplify the presentation we shall here merely assume that interface definitions for all contracts exist and are added to any type environment $\Gamma$ that we shall consider in the following.

We can now give the type rules, starting with the rules for environment agreement, given in Figure~\ref{fig:typerules_env}.
The type judgments here ascertain that each contract, field, variable and method indeed has a type in the type environment, that these types are compatible with the values assigned to them, and, in the case of methods, with the number of steps inferred for the method body $S$.
Note also that we write the type environment as split into two parts: the first part, $\Gamma$, contains the types for contract addresses; the second part, $\Delta$, records the types of local variables, including the `magic' variables \code{this}, \code{sender} and \code{value}.
Also, unless otherwise noted, we shall assume in the rules that these reserved names are contained in the respective sets of field and variables names, i.e.\@ $\code{balance} \in \FNAMES$ and $\code{this}, \code{sender}, \code{value} \in \VNAMES$. 

The type rules for declarations are very similar to those for environments (environments are just an alternative representation of declarations); thus, we omit them from the body of the paper and report them in Appendix~\ref{app:omitted_type_rules}.

\begin{figure}[t]
\centering
\normalfont
\begin{semantics}
  \RULE[t-env-t][ts2_type_termination_envt]({ \Delta = \code{this} : \Gamma(X) })
    { \Gamma, \Delta \vdash \ENV{M} \AND \Gamma \vdash \ENV{T} }
    { \Gamma \vdash (X, \ENV{M}) , \ENV{T} }

  \RULE[t-env-m][ts2_type_termination_envm]
  ({
    \begin{array}{l}
      \Delta(\code{this}) = I                \\
      \Gamma(I)(f) = \TPROC{\VEC{B}}^u_\ell : n \\
      \Delta' = \Delta, \VEC{x} : \VEC{B}, \code{value} : \TINT^u_\ell, \code{sender} : \ITOP 
    \end{array}
  })
    { 
      \Gamma, \Delta' \vdash S : n \AND  
      \Gamma, \Delta  \vdash \ENV{M} 
    }
    { \Gamma, \Delta \vdash (f, (\VEC{x}, S)),\ENV{M} }

  \RULE[t-env-s][ts2_type_termination_envs]({ \Delta = \code{this} : \Gamma(X) })
    { \Gamma, \Delta \vdash \ENV{F} \AND \Gamma \vdash \ENV{S} }
    { \Gamma \vdash (X, \ENV{F}),\ENV{S} }

  \RULE[t-env-f][ts2_type_termination_envf]
  ({
    \begin{array}{l}
      \Delta(\code{this}) = I \\
      \Gamma(I)(p) = B
    \end{array}
  })
    { 
      \Gamma, \Delta \vdash v : B \AND 
      \Gamma, \Delta \vdash \ENV{F} 
    }
    { \Gamma, \Delta \vdash (p, v),\ENV{F} }

  \RULE[t-env-v][ts2_type_termination_envv]({ \Delta(x) = B })
    { 
      \Gamma, \Delta \vdash v : B \AND 
      \Gamma, \Delta \vdash \ENV{V} 
    }
    { \Gamma, \Delta \vdash (x, v, B),\ENV{V} }
\end{semantics}
\caption{Type rules for environment agreement.}
\label{fig:typerules_env}
\end{figure}

Next, we have the rules for typing expressions, which are given in Figure~\ref{fig:typerules_expressions}.
Type judgments are here of the form $\Gamma, \Delta \vdash e : B$.
Note that in rule \nameref{ts2_type_termination_val}, when $v$ is an integer value, we set both the upper and lower bounds in the type to be exactly the value $v$. To widen the type of the value, e.g., when typing an assignment (such as \code{$\TINT^5_2$ $x$ := $3$}), we will use the subtyping relation.
In the rule \nameref{ts2_type_termination_op}, we assume the existence of type rules with judgments of the form $\vdash \op : \VEC{B} \to B$ for each operation $\op$.
In particular, we assume that we can derive conclusions about the upper and lower bounds on integer-valued operations from the bounds on the expression operands.
This is straightforward for standard integer operations. 
As an example, consider the following rules:
\begin{equation*}
  \vdash + : (\TINT^{u_1}_{\ell_1}, \TINT^{u_2}_{\ell_2}) \to \TINT^{u_1 + u_2}_{\ell_1 + \ell_2}
  \qquad
  \vdash - : (\TINT^{u_1}_{\ell_1}, \TINT^{u_2}_{\ell_2}) \to \TINT^{u_1 - \ell_2}_{\ell_1 - u_2}
\end{equation*}

\noindent For example, assuming that $\Gamma, \Delta \vdash x : \TINT^5_2$, the type of $10 - x$ is $\TINT^8_5$.

\begin{figure}[t]
\centering
\normalfont
\begin{semantics}
  \RULE[t-var][ts2_type_termination_var]({ \Delta(x) = B })
    { }
    { \Gamma, \Delta \vdash x : B }
\end{semantics}
\begin{semantics}
  \RULE[t-field][ts2_type_termination_field]({ \Gamma(I)(p) = B })
    { \Gamma, \Delta \vdash e : I }
    { \Gamma, \Delta \vdash e.p : B }
\end{semantics}
\vskip1em
\begin{semantics}
  \RULE[t-val][ts2_type_termination_val]({
  B = %
  \begin{cases}
    \Gamma(v) & \text{if $v \in \ANAMES$}    \\
    \TBOOL    & \text{if $v \in \mathbb{B}$} \\
    \TINT^v_v & \text{if $v \in \mathbb{Z}$}
  \end{cases}
  }) 
    { }
    { \Gamma, \Delta \vdash v : B }

  \RULE[t-op][ts2_type_termination_op]
    { \Gamma, \Delta \vdash \VEC{e} : \VEC{B} \AND \vdash \op : \VEC{B} \to B }
    { \Gamma, \Delta \vdash \op(\VEC{e}) : B }
\end{semantics}
\caption{Type rules for expressions}
\label{fig:typerules_expressions}
\end{figure}

Finally, we have the type rules for statements, given in Figure~\ref{fig:typerules_statements}.
Here, type judgments are of the form $\Gamma, \Delta \vdash S : n$, with $n$ denoting the maximum number of steps required to execute $S$. This is mostly calculated in an obvious way, by only using the upper bounds $u$; however, the lower bounds $\ell$ are implicitly used when calculating the types of expressions (as discussed above). 
For example, we calculate the upper bound on the number of steps of \code{for $e$ do $S$} by first calculating the upper bound $u$ on the value of $e$ and the upper bound $n$ on the steps for $S$: 
If $u \geq 1$, we multiply $u$ and $n+1$ (the ‘+1' is the extra step needed for activating every iteration of the loop, cf.\@ \nameref{ts2_sss_fortrue}) and finally we sum 1 (the cost of the final activation of the loop; i.e., when the guard becomes < 1, see \nameref{ts2_sss_forfalse}). 
Otherwise the upper bound is $1$ (only the case with a guard < 1). 
The max is used for covering the two cases with one single rule.

Notice that \nameref{ts2_type_termination_decv} and \nameref{ts2_type_termination_call} both have $n+2$ in their conclusion, rather than $n+1$, because both constructs will push an extra symbol onto the stack (viz., $\DEL{x}$ and $\ENV{V}$, resp.), which will require an extra step in the semantics when their scope ends. 
A similar reason justifies the ‘+1' in rules \nameref{ts2_type_termination_seq}, \nameref{ts2_type_termination_if}, and \nameref{ts2_type_termination_loop}.

\begin{figure}[t]
\centering
\normalfont
\begin{minipage}{.35\textwidth}
\begin{semantics}
  \RULE[t-skip][ts2_type_termination_skip]
    { }
    { \Gamma, \Delta \vdash \code{skip} : 1 }

  \RULE[t-throw][ts2_type_termination_throw]
    { }
    { \Gamma, \Delta \vdash \code{throw} : 1 }
\end{semantics}
\end{minipage}
\begin{minipage}{.5\textwidth}
\begin{semantics}
  \RULE[t-assv][ts2_type_termination_assv]
    { 
      \Gamma, \Delta \vdash x : B \AND 
      \Gamma, \Delta \vdash e : B 
    }
    { \Gamma, \Delta \vdash \code{$x$ := $e$} : 1 }

  \RULE[t-assf][ts2_type_termination_assf]
    { 
      \Gamma, \Delta \vdash e_1.p : B \AND 
      \Gamma, \Delta \vdash e_2 : B 
    }
    { \Gamma, \Delta \vdash \code{$e_1$.$p$ := $e_2$} : 1 }
\end{semantics}
\end{minipage}

\vspace*{.4cm}

\begin{semantics}
  \RULE[t-decv][ts2_type_termination_decv]
    { 
      \Gamma, \Delta \vdash e : B \AND 
      \Gamma, (\Delta, x : B) \vdash S : n 
    }
    { \Gamma, \Delta \vdash \code{$B$ $x$ := $e$ in $S$} : n + 2 }

  \RULE[t-seq][ts2_type_termination_seq]
    { 
      \Gamma, \Delta \vdash S_1 : n_1 \AND 
      \Gamma, \Delta \vdash S_2 : n_2 
    }
    { \Gamma, \Delta \vdash \code{$S_1$;$S_2$} : n_1 + n_2 + 1}

  \RULE[t-if][ts2_type_termination_if]
    { 
      \Gamma, \Delta \vdash e : \TBOOL \AND 
      \Gamma, \Delta \vdash S_{\TRUE} : n_1 \AND 
      \Gamma, \Delta \vdash S_{\FALSE} : n_2 
    }
    { \Gamma, \Delta \vdash \code{if $e$ then $S_{\TRUE}$ else $S_{\FALSE}$} : \max(n_1, n_2) + 1 }

  \RULE[t-loop][ts2_type_termination_loop]
    { 
      \Gamma, \Delta \vdash e : \TINT^u_\ell \AND 
      \Gamma, \Delta \vdash S : n 
    }
    { \Gamma, \Delta \vdash \code{for $e$ do $S$} : \max(1, u (n + 1)+1) }

  \RULE[t-call][ts2_type_termination_call]
  ({
    \Gamma(I)(f) = \TPROC{\VEC{B}}^u_\ell : n
  })
    { 
      \Gamma, \Delta \vdash e_1 : I \AND
      \Gamma, \Delta \vdash \VEC{e} : \VEC{B} \AND 
      \Gamma, \Delta \vdash e_2 : \TINT^u_\ell
    }
    { \Gamma, \Delta \vdash \code{$e_1$.$f$($\VEC{e}$):$e_2$} : n + 2 }
\end{semantics}
\caption{Type rules for statements.}
\label{fig:typerules_statements}
\end{figure}

We illustrate the use of the type system with a small example. 
Consider the statement \code{for $x$ do call $y$.$f$($x$):1}, and assume it is to be executed in an environment where the local variables $x$ and $y$ are declared so that $\ENV{V}(x) = 3$ and $\ENV{V}(y) = A$, for some contract $A$ of interface type $I$ containing a method $f$.
Thus, this statement will call the method $A$.$f$ three times, each time transferring $1$ unit of currency along with the call, and with the number $3$ as the actual parameter (since the value of $x$ does not change, cf.\@ rule \nameref{ts2_sss_fortrue}).

To type the example, assume further that $\Delta(x) = \TINT^5_1$, $\Delta(y) = I$, and $\Gamma(I)(f) = \TPROC{\TINT}^{10}_1 : 20$.
Thus, the \code{for}-loop can run at most 5 times; moreover, the method call must at least receive 1 unit of currency to execute and will finish in at most 20 steps.
The most interesting part of the typing derivation is then as follows:
\begin{equation*}
\dfrac
{
  \dfrac
  { \Delta(x) = \TINT^5_1 }
  { \Gamma, \Delta \vdash x : \TINT^5_1 } 
  \qquad
  \dfrac
  { 
    \begin{array}{c}
      \tabularnewline
      \ldots 
    \end{array}
    \quad
    \dfrac
    { 
      \dfrac
      { \Delta(x) = \TINT^5_1 }
      { \Gamma, \Delta \vdash x:\TINT^5_1 } 
      \qquad 
      \begin{array}{c}
        \tabularnewline
        \Gamma \vdash \TINT^5_1 \SUBS \TINT 
      \end{array} 
    }
    { \Gamma, \Delta \vdash x:\TINT } 
    \qquad 
    \dfrac
    { \ldots }
    { \Gamma, \Delta \vdash 1:\TINT^{10}_1 }
  }
  { \Gamma, \Delta \vdash \code{call $y$.$f$($x$):1} : 20+2 }
} 
{ \Gamma, \Delta \vdash \code{for $x$ do call $y$.$f$($x$):1} : \max(1,5((20+2)+1)+1) }
\end{equation*}

For space reasons, we have omitted the derivation of $\Gamma, \Delta \vdash y : I$, which is concluded by \nameref{ts2_type_termination_var} and of $\Gamma, \Delta \vdash 1 : \TINT^{10}_1$, which is inferred using the subtyping rules (see Appendix~\ref{app:subtyping}).
The latter involves concluding $\Gamma, \Delta \vdash 1 : \TINT^1_1$ by \nameref{ts2_type_termination_val}, and then $\Gamma \vdash \TINT^1_1 \SUBS \TINT^{10}_1$ can be proved by \nameref{ts2_type_termination_subs_int1}, since $1 \leq 10$ for the upper bound, and $1 \leq 1$ for the lower bound.
Likewise, the type of $x$, which is $\TINT^5_1$, is coerced up to $\TINT$ by the subtyping rule \nameref{ts2_type_termination_subs_int2}, when $x$ is used as a parameter for the method call, as required by the type of $f$.
Thus we conclude, by \nameref{ts2_type_termination_call}, that the method call itself will require at most $20+2 = 22$ steps, and, finally, by \nameref{ts2_type_termination_loop}, that the entire \code{for}-loop will require at most $\max(1,5((20+2)+1)+1) = 116$ steps. 

\subsection{Properties of the type system}\label{sec:prop}
Our type system ensures that a well-typed statement $S$, executing in a configuration with well-typed environments $\ENV{TSV}$, will terminate after at most $n$ steps.
To prove this property, we wish to show a statement saying that well-typedness is preserved by transitions (subject reduction), and furthermore that the number $n$ decreases after every transition step.
However, our semantics does not describe the execution of single statements, but rather of \emph{stacks}, so we shall also need the notion of well-typed stacks and configurations.

The definition of such notions is complicated by the form of our semantic rules (Figure~\ref{fig:tinysolgas_semantics_statements1_sss}):
Atomic statements are taken from the top of the stack and evaluated directly in the conclusion of a rule, whereas composite statements are rewritten and then pushed back onto the stack.
This feature makes exception handling easy, since we do not need extra rules for propagating an exception down through a derivation tree, but it does pose a problem w.r.t.\@ the subject reduction theorem, in particular in dealing with the rules \nameref{ts2_sss_decv} and \nameref{ts2_sss_call}.
Both of these rules alter the local variable environment $\ENV{V}$, either by adding a new binding to it (\nameref{ts2_sss_decv}) or by replacing it with a new environment (\nameref{ts2_sss_call}), but the scope of these changes is not confined to the premise of those rules, unlike in the corresponding type rules (\nameref{ts2_type_termination_decv} and \nameref{ts2_type_termination_call}).
Instead, the syntax of stacks (Definition~\ref{def:stack_syntax}) introduces some extra symbols ($\DEL{x}$ and $\ENV{V}$) to mark the end of these scopes, which are then pushed onto the stack \emph{before} pushing the statement $S$ to be executed within the scope of these new bindings.
%\footnote{This is also the reason why \nameref{ts2_type_termination_decv} and \nameref{ts2_type_termination_call} both have $n+2$ in their conclusion, rather than $n+1$, because both constructs will push an extra symbol onto the stack ($\DEL{x}$ resp.\@ $\ENV{V}$), which will require an extra step to handle, when their scopes end.}. 
This, in turn, means that new local variables can appear as free in the statement residing on the top of a stack \emph{after} a transition step, if the previously executed statement was a declaration or a call. 
Thus, to be able to type the stack after a transition step, we need a way to  update the $\Delta$ part of the type environment with information on these new bindings. 
This is done as follows: 

\begin{definition}[Extraction function]\label{def:extraction_function}
The type extraction function from stacks to type environments is given by the following clauses:
\normalfont
\begin{align}
  \label{extract1} \EXTRACT{(x,v,B), \ENV{V} :: Q}                  & = x:B ,\ \EXTRACT{\ENV{V} :: Q} \\
  \label{extract2} \EXTRACT{\ENV{V}^\emptyset :: Q}                 & = \epsilon                      \\
  \label{extract3} \EXTRACT{\code{$B$ $x$ := $e$ in $S$} :: Q}      & = x:B,\ \Delta                  \\
  \label{extract4} \EXTRACT[\Delta, x:B]{\DEL{x} :: Q}              & = \Delta                        \\
  \label{extract5} \EXTRACT{\code{$e_1$.$f$($\VEC{e}$):$e_2$} :: Q} & = \code{this}:I,\ \code{sender}:\ITOP, \ \code{value}:\TINT^u_\ell,\ \VEC{x}:\VEC{B}           \\
  \notag                                                            & \text{where $\Gamma, \Delta \vdash e_1 : I$ and $\Gamma(I)(f) = \TPROC{\VEC{B}}^u_\ell : n$.}  \\
  \label{extract6} \EXTRACT{Q}                                      & = \Delta \qquad \text{in all other cases}
\end{align}
\end{definition}
% TODO: How do we get the names $\VEC{x}$?

The function ensures that the type environment for local variables $\Delta$ remains in agreement with the currently active variable environment $\ENV{V}$ after a transition.
This is also the reason why we needed to store the type information of locally declared variables in $\ENV{V}$, and why we use the convention of writing the type environment as split into two segments (viz., $\Gamma$ and $\Delta$), since the $\Gamma$ segment remains unaltered.

\begin{figure}[t]
\begin{minipage}{.4\textwidth}
\begin{semantics}
  \RULE[t-bot][ts2_type_termination_bot]
    { }
    { \Gamma, \Delta \vdash \bot : 0 }

  \RULE[t-exc][ts2_type_termination_exc]
    { }
    { \Gamma, \Delta \vdash \EXC{l} :: Q : 0  }

  \RULE[t-del][ts2_type_termination_del]
    { \Gamma, \EXTRACT{\DEL{x} :: Q} \vdash Q : n }
    { \Gamma, \Delta \vdash \DEL{x} :: Q : n }
\end{semantics}
\end{minipage}
\begin{minipage}{.5\textwidth}
\begin{semantics}
  \RULE[t-stm][ts2_type_termination_stm]
    { \Gamma, \Delta \vdash S : n_1 \AND \Gamma, \Delta \vdash Q : n_2 }
    { \Gamma, \Delta \vdash S :: Q : n_1 + n_2 }

  \RULE[t-ctx][ts2_type_termination_ctx]
    { \Gamma, \EXTRACT{\ENV{V} :: Q} \vdash Q : n }
    { \Gamma, \Delta \vdash \ENV{V} :: Q : n }

  \RULE[t-cfg][ts2_type_termination_cfg](n < g)
    { \Gamma, \Delta \vdash Q : n \AND \Gamma \vdash \ENV{S} \AND \Gamma, \Delta \vdash \ENV{V} }
    { \Gamma, \Delta \vdash \CONF{Q, \ENV{SV}, g} }
\end{semantics}
\end{minipage}
\caption{Type rules for stacks and configurations.}
\label{fig:typerules_stacks}
\end{figure}

Now, we can define the notions of well-typed stacks and configurations, which result from the rules given in Figure~\ref{fig:typerules_stacks}.
Note that we use the extraction function in the premises of rules \nameref{ts2_type_termination_ctx} and \nameref{ts2_type_termination_del} to handle the two special cases when the top of the stack is an end-of-scope symbol.
Moreover, rule \nameref{ts2_type_termination_cfg} states that a configuration $\CONF{Q, \ENV{SV}, g}$ is well-typed (for some type environment $\Gamma, \Delta$) if the environments agree, and the stack is typable as terminating in $n$ steps, \emph{and} the supplied gas value $g$ is greater than $n$; this is precisely what we want our type system to ensure.

We can now state our main theorem, whose proof is in Appendix~\ref{app:proof_subject_reduction}.

\begin{theorem}[Subject reduction]\label{thm:subject_reduction}
Let $\Gamma \vdash \ENV{T}$ and $\Gamma, \Delta \vdash \CONF{Q, \ENV{SV}, g}$. If $\ENV{T} \vdash \CONF{Q, \ENV{SV}, g} \trans \CONF{Q', \ENV{SV}', g'}$ then $\Gamma, \EXTRACT{Q} \vdash \CONF{Q', \ENV{SV}', g'}$.
\end{theorem}

This theorem ensures not only that well-typedness of configurations is preserved, but also that the number of steps always remains lower than the provided gas value, by the side condition in rule \nameref{ts2_type_termination_cfg}.

\begin{corollary}
Let $\Gamma \vdash \ENV{T}$ and $\Gamma, \Delta \vdash \CONF{Q, \ENV{SV}, g}$. If $\ENV{T} \vdash \CONF{Q, \ENV{SV}, g} \trans \CONF{Q', \ENV{SV}', g'}$ then $Q' \neq \EXC{\code{oog}} :: Q''$.
\end{corollary}

It is perhaps also worth emphasising that the type system also ensures termination.
This is another simple consequence of Theorem~\ref{thm:subject_reduction}, since the gas value decreases at every step, except when the transition is concluded with one of the rules \nameref{ts2_sss_seq}, \nameref{ts2_sss_delv}, \nameref{ts2_sss_return}, \nameref{ts2_sss_throw} and \nameref{ts2_sss_oog}.
Of these, the last two place an exception on the top of the stack, which has no further transitions, so the execution terminates.
Since all stacks have a finite depth, the first three rules can only be used finitely many times before another, gas-consuming statement is encountered, or we reach the bottom of the stack.
Thus, the number of steps $n$ must necessarily also always eventually decrease.

Lastly, Theorem~\ref{thm:subject_reduction} can be generalised to transactions and blockchains in an obvious way, since a transaction simply corresponds to a method call placed on an empty stack (see rule \nameref{ts2_trans_trans1}), from which a $\Delta$ is derivable using the extraction function from Definition~\ref{def:extraction_function}.
Given well-typed environments $\ENV{TS}$, the type system can thus be used to ensure that a transaction is supplied with sufficient gas to allow it to finish without invoking rule \nameref{ts2_trans_trans2}.

\subsection{Limitations of the type system}\label{sec:lim}
The type system ensures termination of well-typed configurations (and hence transactions) and absence of out-of-gas exceptions, but at the cost of some limitations.
Firstly, we had to abandon \code{while} loops, in favour of the (more controllable) \code{for} loops.
This does not seem too severe, since smart contracts are not intended to be used to create non-terminating programs.

More limiting, perhaps, is the fact that it also prohibits all forms of recursive calls, including mutual recursions.
This is a consequence of rule \nameref{ts2_type_termination_call}, which assigns a number of $n+2$ steps to a call to a method $f$, whose number itself is $n$.
By rule \nameref{ts2_type_termination_envm}, the declaration of $f$ is well-typed if the statement $S$ in the method body is typeable as terminating in $n$ steps.
As $n$ is strictly less than $n+2$, $S$ therefore cannot contain a recursive call to $f$, nor a call to any other method that itself calls $f$.
This limitation seems more severe, since it rules out even terminating recursive functions. 

Finally, the bounded integer type needed for typing the guards of \code{for} loops does not permit some common assignments where the assigned variable appears in the expression (e.g., \code{$x$ := $x$+1}).
This limitation is a consequence of rules \nameref{ts2_type_termination_assv}, \nameref{ts2_type_termination_decv} and \nameref{ts2_type_termination_assf}, which require the assigned variable (resp.\@ field) to have the same base type $B$ as the expression $e$.
Suppose for example that $x$ has type $\TINT^u_\ell$, and we have the assignment \code{$x$ := $x$+1}.
The computed bounds for the expression (which we assume in the rule \nameref{ts2_type_termination_op}) must therefore be $\TINT^{u+1}_{\ell+1}$, which does not match the type of $x$, nor can it be coerced down to match it via subtyping.
However, we remark that this limitation only affects those variables that appear in the guard of \code{for} loops, since all other ones can be typed as $\TINT$. 

\section{Conclusion and Future Work}\label{sec:concl}
In this paper, we have further extended the \TINYSOL{} model by equipping it with a small-step semantics, exceptions and a gas mechanism.
Using this refined smart-contract model language, we presented a first step towards the development of static analysis techniques for checking the desirable property that a transaction never runs out-of-gas during its execution. 
This is done by means of a type system whose main aim is to give an upper bound on the number of steps needed by a transaction to complete its task.

The present type system is, in particular, inspired by Type System 1 of \cite{PITERM}. That paper also presents two other type systems, which e.g.\@ allow for some limited recursion.
Thus, one avenue for future work is to extend the present type system in a similar manner to permit some recursive method calls.

Another avenue concerns an altogether different approach to type soundness.
The present work uses the usual, inductive (subject reduction) approach of \cite{wright_felleisen1994subject_reduction}.
However, one way to lift some of the restrictions on the loop construct would be to instead use a coinductive approach, sometimes known as \emph{semantic typing} \cite{ahmed2004semantic_types_phd,appel_mcallester2001semantic_types_pcc,caires2007}.
A key difference lies in the structure of the soundness proof, which in semantic typing allows `manual proofs' to be used to guarantee well-typedness of some `unsafe' pieces of code, which cannot otherwise be judged type-safe using the type rules.
This approach has in particular been explored by  \cite{ahmed2004semantic_types_phd,appel_mcallester2001semantic_types_pcc} in a setting of proof-carrying code \cite{necula/1997/popl/pcc}, the point being that such proofs can be automatically verified by a proof-checker.
A `manual proof' could for example be used to show that a particular usage of a \code{while} loop in a contract in fact \emph{will} terminate, and it could then be published on the blockchain along with the contract, thereby taking advantage of the immutable nature of blockchains.
This would allow users to still type check transactions, even if they invoke methods containing unsafe code such as \code{while} loops, because the type checker can take the accompanying `manual proof' into account.
We intend to explore this possibility in future work.

An orthogonal way to evolve our type system is to let it also provide a lower bound on the gas needed to execute a piece of code (not just an upper bound, as it is now). 
This can be useful for estimating how much currency one needs to have in the account to invoke a certain functionality.

Finally, having introduced the gas mechanism, we plan to implement real-life smart contracts in \TINYSOL{} and apply our type system to them to check its practical applicability to meaningful case studies.

\bibliographystyle{splncs04}
\bibliography{literature}

\begin{thebibliography}{10}
\providecommand{\url}[1]{\texttt{#1}}
\providecommand{\urlprefix}{URL }
\providecommand{\doi}[1]{https://doi.org/#1}

\bibitem{AGL24}
Aceto, L., Gorla, D., Lybech, S.: A sound type system for secure currency flow.
  In: Proc. of {ECOOP} 2024. LIPICS vol. 313, 19:1 – 19:27 (2024),
  \url{https://arxiv.org/abs/2405.12976}

\bibitem{ahmed2004semantic_types_phd}
Ahmed, A.J.: Semantics of Types for Mutable State. Ph.D. thesis, Princeton
  University (2004), \url{http://www.ccs.neu.edu/home/amal/ahmedsthesis.pdf}

\bibitem{appel_mcallester2001semantic_types_pcc}
Appel, A.W., McAllester, D.: An indexed model of recursive types for
  foundational proof-carrying code. ACM Trans. Program. Lang. Syst.
  \textbf{23}(5),  657–683 (Sep 2001). \doi{10.1145/504709.504712}

\bibitem{bartoletti2019tinysol}
Bartoletti, M., Galletta, L., Murgia, M.: A minimal core calculus for
  {Solidity} contracts. In: Data Privacy Management, Cryptocurrencies and
  Blockchain Technology. pp. 233--243. Springer (2019).
  \doi{10.1007/978-3-030-31500-9_15}

\bibitem{Boudol10}
Boudol, G.: Typing termination in a higher-order concurrent imperative
  language. Information and Computation  \textbf{208}(6),  716--736 (2010),
  \url{https://doi.org/10.1016/j.ic.2009.06.007}

\bibitem{caires2007}
Caires, L.: Logical semantics of types for concurrency. In: Proc. of {CALCO}.
  LNCS, vol.~4624, pp. 16--35. Springer (2007),
  \url{https://doi.org/10.1007/978-3-540-73859-6\_2}

\bibitem{CookPR11}
Cook, B., Podelski, A., Rybalchenko, A.: Proving program termination. Commun.
  {ACM}  \textbf{54}(5),  88--98 (2011),
  \url{https://doi.org/10.1145/1941487.1941509}

\bibitem{PITERM}
Deng, Y., Sangiorgi, D.: Ensuring termination by typability. Information and
  Computation  \textbf{204}(7),  1045--1082 (2006).
  \doi{10.1016/j.ic.2006.03.002}

\bibitem{GenetJS20}
Genet, T., Jensen, T.P., Sauvage, J.: Termination of {Ethereum}'s smart
  contracts. In: Proc. of the 17th International Joint Conference on e-Business
  and Telecommunications - Volume 2: SECRYPT. pp. 39--51. ScitePress (2020),
  \url{https://doi.org/10.5220/0009564100390051}

\bibitem{GrechKJBSS18}
Grech, N., Kong, M., Jurisevic, A., Brent, L., Scholz, B., Smaragdakis, Y.:
  {MadMax}: Surviving out-of-gas conditions in {Ethereum} smart contracts.
  Proc. {ACM} Program. Lang.  \textbf{2}({OOPSLA}),  116:1--116:27 (2018),
  \url{https://doi.org/10.1145/3276486}

\bibitem{grishchenko2018}
Grishchenko, I., Maffei, M., Schneidewind, C.: A semantic framework for the
  security analysis of {Ethereum} smart contracts. In: Proc. of {POST}. LNCS,
  vol. 10804, pp. 243--269. Springer (2018),
  \url{https://doi.org/10.1007/978-3-319-89722-6\_10}

\bibitem{necula/1997/popl/pcc}
Necula, G.C.: Proof-carrying code. In: Proc. of {POPL}. p. 106–119. ACM
  (1997), \url{https://doi.org/10.1145/263699.263712}

\bibitem{Nielson87}
Nielson, H.R.: A {Hoare}-like proof system for analysing the computation time
  of programs. Sci. Comput. Program.  \textbf{9}(2),  107--136 (1987),
  \url{https://doi.org/10.1016/0167-6423(87)90029-3}

\bibitem{NielsonN1992}
Nielson, H.R., Nielson, F.: Semantics with Applications - A Formal
  Introduction. Wiley Professional Computing, Wiley (1992)

\bibitem{EVM}
Wood, G.: Ethereum: A secure decentralised generalised transaction ledger. In:
  Yellow Paper (2014), updated for EIP-150 in 2017

\bibitem{wright_felleisen1994subject_reduction}
Wright, A.K., Felleisen, M.: A syntactic approach to type soundness.
  Information and Computation  \textbf{115}(1),  38--94 (1994)

\end{thebibliography}

\appendix
\section{Omitted Semantic Rules}\label{app:omitted_semantics_rules}
The semantics of expressions is given in the following Figure~\ref{fig:tinysolgas_semantics_expressions}, where we assume an evaluation function $\op(\VEC{v}) \trans_{\op} v$ for every basic operation $\op$ of the language.

\begin{figure}\centering
\begin{semantics}
  \RULE[exp-val][ts2_expr_val]
    { }
    { \ENV{SV} \vdash v \trans_e v }

  \RULE[exp-var][ts2_expr_var]
    { x \in \DOM{\ENV{V}} \AND
      \ENV{V}(x) = v}
    { \ENV{SV} \vdash x \trans_e v }

  \RULE[exp-field][ts2_expr_field]
    { \ENV{SV} \vdash e \trans_e X \AND p \in \DOM{\ENV{S}(X)} \AND
      \ENV{S}(X)(p) = v}
    { \ENV{SV} \vdash e.p \trans_e v }

  \RULE[exp-op][ts2_expr_op]
    { \ENV{SV} \vdash \VEC{e} \trans_e \VEC{v} \AND \op(\VEC{v}) \trans_{\op} v}
    { \ENV{SV} \vdash \op(\VEC{e}) \trans_e v } 
\end{semantics}
\caption{Semantics of expressions.}
\label{fig:tinysolgas_semantics_expressions}
\end{figure}

The semantics of declarations turns a field/method/contract declaration in the appropriate environment, to be used in the semantics of transactions.
The rules are given in the following Figure~\ref{fig:tinysolgas_semantics_declarations}

\begin{figure}\centering
\begin{semantics}
  \RULE[dec-f$_1$]
    { }
    { \CONF{\epsilon, \ENV{F}} \trans_{DF} \ENV{F} }

  \RULE[dec-f$_2$]
    { \CONF{DF, \ENV{F}} \trans_{DF} \ENV{F}' }
    { \CONF{\code{field $p$ := $v$;$DF$}, \ENV{F}} \trans_{DF} (p, v), \ENV{F}' }

  \RULE[dec-m$_1$]
    { }
    { \CONF{\epsilon, \ENV{M}} \trans_{DM} \ENV{M} }

  \RULE[dec-m$_2$]
    { \CONF{DM, \ENV{M}} \trans_{DM} \ENV{M}' }
    { \CONF{\code{$f$($\VEC{x}$) \{ $S$ \} $DM$}, \ENV{M}} \trans_{DM} (f, (\VEC{x}, S)), \ENV{M}' }

  \RULE[dec-c$_1$][ts_dec_c1]
    { }
    { \CONF{\epsilon, \ENV{ST}} \trans_{DC} \ENV{ST} }

  \RULE[dec-c$_2$][ts_dec_c2]
    {
        \CONF{DF, \ENV{F}^{\EMPTYSET}} \trans_{DF} \ENV{F} \AND
        \CONF{DM, \ENV{M}^{\EMPTYSET}} \trans_{DM} \ENV{M} \AND
        \CONF{DC, \ENV{ST}} \trans_{DC} \ENV{ST}'  
    }
    { \CONF{\code{contract $X$ \{ $DF$ $DM$ \} $DC$}, \ENV{ST}} \trans_{DC} ((X, \ENV{F}), \ENV{S}'), ((X, \ENV{M}), \ENV{T}') }
\end{semantics}
\caption{Semantics of declarations.}
\label{fig:tinysolgas_semantics_declarations}
\end{figure}

\section{Subtyping rules}\label{app:subtyping}
The subtyping rules are only used for the right-hand side expression of assignments and for the arguments to method calls, to allow the value of the expression to be coerced up to a less specific type that can match the type of the variable or field.
Thus, we need the usual subsumption rule for expressions:
\begin{center}
\begin{semantics}
  \RULE[t-subs]
    { \Gamma, \Delta \vdash e : B_1 \AND \Gamma \vdash B_1 \SUBS B_2 }
    { \Gamma, \Delta \vdash e : B_2 }
\end{semantics}
\end{center}

Note that the subtyping relation is parametrised with a type environment $\Gamma$.
This is necessary for subtyping of interface names $I$.
The subtyping relation is then given by the reflexive and transitive closure of the rules in Figure~\ref{fig:subtyping}.

\begin{figure}
\begin{semantics}
  \RULE[subs-int$_1$][ts2_type_termination_subs_int1]
  ({
    \begin{array}{c}
      u_1 \leq u_2       \\
      \ell_1 \geq \ell_2
    \end{array}
  })
    { }
    { \Gamma \vdash \TINT^{u_1}_{\ell_1} \SUBS \TINT^{u_2}_{\ell_2} }

  \RULE[subs-int$_2$][ts2_type_termination_subs_int2]
    { }
    { \Gamma \vdash \TINT^u_\ell \SUBS \TINT }

  \RULE[subs-name][ts2_type_termination_subs_name]
    { \Gamma \vdash \Gamma(I^1) \SUBS \Gamma(I^2) }
    { \Gamma \vdash I^1 \SUBS I^2 }

  \RULE[subs-env][ts2_type_termination_subs_env]
    { \forall n \in \DOM{\Gamma_2} \SUCHTHAT \Gamma_1(n) \SUBS \Gamma_2(n) }
    { \Gamma \vdash \Gamma_1 \SUBS \Gamma_2 }

  \RULE[subs-proc][ts2_type_termination_subs_proc]
   ({
    \begin{array}{c}
      n_1 \leq n_2       \\
      u_1 \leq u_2       \\
      \ell_1 \geq \ell_2
    \end{array}
  })
    { \Gamma \vdash \VEC{B}_1 \SUBS \VEC{B}_2 }
    { \Gamma \vdash \TPROC{\VEC{B}_1}^{u_1}_{\ell_1}:n_1 \SUBS \TPROC{\VEC{B}_2}^{u_2}_{\ell_2}:n_2 }
\end{semantics}
\caption{Subtyping rules.}
\label{fig:subtyping}
\end{figure}

Note in particular the two rules for subtyping of the two integer types.
The first rule, \nameref{ts2_type_termination_subs_int1}, allows us to widen the bounds on a bounded integer.
This is necessary to make statements such as \code{$\TINT^8_2$ $x$ := 5} be well-typed, since the type of \code{5} is $\TINT^5_5$ (by rule \nameref{ts2_type_termination_val}).
The second rule, \nameref{ts2_type_termination_subs_int2},  allows any bounded integer type to be coerced up to an unbounded integer type.
This is necessary, if e.g.\@ both bounded and unbounded integers appear together in an expression, or if we have an assignment of a bounded integer type to an unbounded integer variable.

\section{Omitted type rules}\label{app:omitted_type_rules}
The type rules for declarations are given in the following Figure~\ref{fig:typerules_declarations}.
Note that they are almost identical to the type rules for environments (Figure~\ref{fig:typerules_env}).

\begin{figure}[h]
\centering
\normalfont
\begin{semantics}
  \RULE[t-dec-c][ts2_type_termination_decc]({ \Delta = \code{this} : \Gamma(X) })
    { 
      \Gamma, \Delta \vdash DF \AND
      \Gamma, \Delta \vdash DM  \AND
      \Gamma  \vdash DC 
    }
    { \Gamma \vdash \code{contract $X$ \{ $DF$ $DM$ \} $DC$} }

  \RULE[t-dec-f][ts2_type_termination_decf]
  ({
    \begin{array}{l}
      \Delta(\code{this}) = I \\
      \Gamma(I)(p) = B
    \end{array}
  })
    { 
      \Gamma, \Delta \vdash v : B \AND 
      \Gamma, \Delta \vdash DF }
    { \Gamma, \Delta \vdash \code{field $p$ := $v$; $DF$} }

  \RULE[t-dec-m][ts2_type_termination_decm]
  ({
    \begin{array}{l}
      \Delta(\code{this}) = I                \\
      \Gamma(I)(f) = \TPROC{\VEC{B}}^u_\ell : n \\
      \Delta' = \Delta, \VEC{x} : \VEC{B}, \code{value} : \TINT^u_\ell, \code{sender} : \ITOP
    \end{array} 
  })
    { \Gamma, \Delta' \vdash S : n \AND \Gamma, \Delta \vdash DM }
    { \Gamma, \Delta \vdash \code{$f(\VEC{x})$ \{ $S$ \} $DM$} }
\end{semantics}
\caption{Type rules for declarations.}
\label{fig:typerules_declarations}
\end{figure}

\section{Proof of Theorem~\ref{thm:subject_reduction} (subject reduction)}\label{app:proof_subject_reduction}
The proof relies on the following standard lemmas. 

\begin{lemma}[Strengthening]\label{lemma:envv_strengthening}
If $\Gamma, (\Delta, x:B) \vdash \ENV{V}$ and $x \notin \DOM{\ENV{V}}$, then $\Gamma, \Delta \vdash \ENV{V}$.  
\end{lemma}
\begin{proof}
By induction on the structure of $\ENV{V}$, by using rule \nameref{ts2_type_termination_envv}.    
\end{proof}

The next two lemmas simply state that well-typedness of the variable environment (resp.\@ the state) is preserved, if we replace a value $v_1$ with another value $v_2$ of the same type as $v_1$.
Both are shown by induction on the structure of $\ENV{V}$ (resp.\@ $\ENV{S}$ and $\ENV{P}$), by using the rules \nameref{ts2_type_termination_envv} resp.\@ \nameref{ts2_type_termination_envs} and \nameref{ts2_type_termination_envf}.

\begin{lemma}[Update for variables]\label{lemma:update_variables}
Assume that $\Gamma, \Delta \vdash \ENV{V}$, $x \in \DOM{\ENV{V}}$, $\Delta(x) = B$ and $\Gamma, \Delta \vdash v : B$. Then $\Gamma, \Delta \vdash \ENV{V}\EXTEND{x}{v}$.
\end{lemma}

\begin{lemma}[Update for fields]\label{lemma:update_fields}
Assume that $\Gamma \vdash \ENV{S}$, $X \in \DOM{\ENV{S}}$,  $\Gamma(X)(p) = B$ and $\Gamma, \Delta \vdash v : B$.
Then $\Gamma \vdash \ENV{S}\EXTEND{X}{\ENV{P}\EXTEND{p}{v}}$.
\end{lemma}

Next, we need a lemma stating that, if $e$ has a type $B$ and $e$ evaluates to a value $v$ relative to well-typed environments $\ENV{SV}$, then $v$ will indeed be a value of type $B$.

\begin{lemma}[Safety for expressions]\label{lemma:safety_expressions}
Assume $\Gamma \vdash \ENV{S}$ and $\Gamma, \Delta \vdash \ENV{V}$ and $\Gamma, \Delta \vdash e : B$.
If $\ENV{SV} \vdash e \trans v$, then $\Gamma, \Delta \vdash v : B$.
\end{lemma}
\begin{proof}
By induction on the structure of $e$.
\end{proof}

We are now ready to prove subject reduction.

%The theorem states the following: If $\Gamma \vdash \ENV{T}$ and $\Gamma, \Delta \vdash \CONF{Q, \ENV{SV}, g}$ and $\ENV{T} \vdash \CONF{Q, \ENV{SV}, g} \trans \CONF{Q', \ENV{SV}', g'}$ then $\Gamma, \EXTRACT{Q} \vdash \CONF{Q', \ENV{SV}', g'}$.

\begin{proof}[of Theorem~\ref{thm:subject_reduction}]
$\Gamma, \Delta \vdash \CONF{Q, \ENV{SV}, g}$ was concluded by \nameref{ts2_type_termination_cfg}; from the premises and side condition of this rule, we know that:
\begin{itemize}
  \item $\Gamma, \Delta \vdash Q : n$, 
  \item $\Gamma \vdash \ENV{S}$, 
  \item $\Gamma, \Delta \vdash \ENV{V}$, and 
  \item $n < g$.
\end{itemize}

\noindent The transition is of the form
\begin{equation}
\label{eq:red}
  \ENV{T} \vdash \CONF{Q, \ENV{SV}, g} \trans \CONF{Q', \ENV{SV}', g'} 
\end{equation} 

\noindent and our goal is therefore to show that $\Gamma, \EXTRACT{Q} \vdash \CONF{Q', \ENV{SV}', g'}$ also holds.
This must again be concluded by \nameref{ts2_type_termination_cfg}; hence, we must show that the premises of this rule, i.e.
\begin{itemize}
  \item $\Gamma, \EXTRACT{Q} \vdash Q' : n'$, 
  \item $\Gamma \vdash \ENV{S}'$, 
  \item $\Gamma, \EXTRACT{Q} \vdash \ENV{V}'$, and 
  \item $n' < g'$
\end{itemize}

\noindent are all satisfied.
We therefore proceed by case analysis on how the transition \eqref{eq:red} has been inferred (by considering the rules in Figure~\ref{fig:tinysolgas_semantics_statements1_sss}):
\begin{itemize}

  %%% SKIP 
  \item Suppose the transition was concluded by \nameref{ts2_sss_skip}.
  Then the stack is of the form $\code{skip} :: Q$, and \eqref{eq:red} is of the form
  \begin{equation*}
    \ENV{T} \vdash \CONF{\code{skip} :: Q, \ENV{SV}, g} \trans \CONF{Q, \ENV{SV}, g-1} 
  \end{equation*}
  Now, $\Gamma, \Delta \vdash \code{skip} :: Q : n$ must have been concluded by \nameref{ts2_type_termination_stm}, since the top element on the stack is a statement.
  The conclusion tells us that $n = n_1 + n_2$, and from the premise we have that
  \begin{align*}
    \Gamma, \Delta & \vdash \code{skip} : n_1 \\
    \Gamma, \Delta & \vdash Q : n_2
  \end{align*}

  Here, $\Gamma, \Delta \vdash \code{skip} : n_1$ must have been concluded by \nameref{ts2_type_termination_skip}, which tells us that $n_1 = 1$.
  Thus $n = 1 + n_2$, and therefore $n' = n_2 = n-1$.

  Finally, $\EXTRACT{\code{skip} :: Q} = \Delta$ by case \eqref{extract6} of Definition~\ref{def:extraction_function}.
  Putting this together, we thus have the following:
  \begin{itemize}
    \item $\Gamma, \Delta \vdash Q : n-1$ as argued above,
    \item $\Gamma \vdash \ENV{S}$ by assumption, since the transition does not alter $\ENV{S}$,
    \item $\Gamma, \Delta \vdash \ENV{V}$ by assumption, since the transition does not alter $\ENV{V}$,
    \item $n' < g'$, since $n' = n-1$, $g' = g-1$, and $n < g$.
  \end{itemize}

%  As all the premises are satisfied, we can thus conclude by \nameref{ts2_type_termination_cfg} that $\Gamma, \Delta \vdash \CONF{Q, \ENV{SV}, g-1}$ indeed holds.
  \vskip1em

  %%% S1;S2
  \item Suppose that \nameref{ts2_sss_seq} was used; then \eqref{eq:red} is
  \begin{equation*}
    \ENV{T} \vdash \CONF{\code{$S_1$;$S_2$} :: Q, \ENV{SV}, g} \trans \CONF{S_1 :: S_2 :: Q, \ENV{SV}, g}
  \end{equation*}
Furthermore, $\Gamma, \Delta \vdash \CONF{\code{$S_1$;$S_2$} :: Q} : n$ was concluded by \nameref{ts2_type_termination_stm}, where $n = n_1 + n_2$, and with premises
  \begin{align*}
    \Gamma, \Delta & \vdash \code{$S_1$;$S_2$} : n_1 \\
    \Gamma, \Delta & \vdash Q : n_2
  \end{align*}

  \noindent of which the first was concluded by \nameref{ts2_type_termination_seq}.
  Thus $n_1 = n_1' + n_1'' + 1$ from the conclusion of that rule, and
  \begin{align*}
    \Gamma, \Delta & \vdash S_1 : n_1' \\
    \Gamma, \Delta & \vdash S_2 : n_1''
  \end{align*}

  \noindent from its premises.

  As $\EXTRACT{\code{$S_1$;$S_2$} :: Q'} = \Delta$ by case \eqref{extract6} of Definition~\ref{def:extraction_function}, we can therefore conclude 
  \begin{align*}
    \Gamma, \Delta & \vdash S_2 :: Q : n_1'' + n_2              \\
    \Gamma, \Delta & \vdash S_1 :: S_2 :: Q : n_1' + n_1'' + n_2
  \end{align*}

  \noindent by repeated application of rule \nameref{ts2_type_termination_stm}.
  Thus we conclude the following:
  \begin{itemize}
    \item $\Gamma, \Delta \vdash S_1 :: S_2 :: Q : n_1' + n_1'' + n_2$ as argued above, 
    \item $\Gamma \vdash \ENV{S}$ by assumption, since the transition does not alter $\ENV{S}$, 
    \item $\Gamma, \Delta \vdash \ENV{V}$ by assumption, since the transition does not alter $\ENV{V}$,
    \item $n_1' + n_1'' + n_2 < g$, since we know that $n_1' + n_1'' + 1 + n_2 < g$, and the transition does not consume any gas.
  \end{itemize}
  \vskip1em

  \item Suppose that \nameref{ts2_sss_if} was used; then \eqref{eq:red} is
  \begin{equation*}
    \ENV{T} \vdash  \CONF{\code{if $e$ then $S_{\TRUE}$ else $S_{\FALSE}$} :: Q, \ENV{SV}, g} \trans \CONF{S_b :: Q, \ENV{SV}, g-1} 
  \end{equation*}
Furthermore, $\Gamma, \Delta \vdash \code{if $e$ then $S_\TRUE$ else $S_\FALSE$} :: Q : n$ was concluded by \nameref{ts2_type_termination_stm}, where $n = n_1 + n_2$, and with premises
  \begin{align*}
    \Gamma, \Delta & \vdash \code{if $e$ then $S_\TRUE$ else $S_\FALSE$} : n_1 \\
    \Gamma, \Delta & \vdash Q : n_2
  \end{align*}

  \noindent of which the first was concluded by \nameref{ts2_type_termination_if}.
  Thus $n_1 = \max(n_\TRUE, n_\FALSE) + 1$ from the conclusion of that rule, and 
  \begin{align*}
    \Gamma, \Delta & \vdash S_\TRUE : n_\TRUE \\
    \Gamma, \Delta & \vdash S_\FALSE : n_\FALSE 
  \end{align*}

  \noindent from its premises.
  Say that $n_b = \max(n_\TRUE, n_\FALSE)$, then $n = n_b + 1 + n_2$.

  Finally, $\EXTRACT{\code{if $e$ then $S_\TRUE$ else $S_\FALSE$} :: Q} = \Delta$ by case \eqref{extract6} of Definition~\ref{def:extraction_function}.
  Thus we can conclude the following:
  \begin{itemize}
    \item $\Gamma, \Delta \vdash S_b :: Q : n_b + n_2$, which we conclude by \nameref{ts2_type_termination_stm}. 
    \item $\Gamma \vdash \ENV{S}$ by assumption, since the transition does not alter $\ENV{S}$.
    \item $\Gamma, \Delta \vdash \ENV{V}$ by assumption, since the transition does not alter $\ENV{V}$.
    \item $n_b + n_2 < g - 1$, since we know that $n_b + 1 + n_2 < g$.
  \end{itemize}
  \vskip1em

  %%% for e do S, true
  \item Suppose that \nameref{ts2_sss_fortrue} was used; then \eqref{eq:red} is
  \begin{equation*}
    \ENV{T} \vdash \CONF{\code{for $e$ do $S$} :: Q, \ENV{SV}, g} \trans \CONF{S :: \code{for $v'$ do $S$} :: Q, \ENV{SV}, g-1} 
  \end{equation*}

  \noindent with $\ENV{SV} \vdash e \trans v$, $v \geq 1$, and $v' = v-1$ as premises.
  Then $\Gamma, \Delta \vdash \code{for $e$ do $S$} :: Q : n$ was concluded by \nameref{ts2_type_termination_stm}, where 
  \begin{equation*}
    n = \max(1, u (n_1' + 1)+1) + n_2
  \end{equation*}

  \noindent and with premises
  \begin{align*}
    \Gamma, \Delta & \vdash \code{for $e$ do $S$} : \max(1, u (n_1' + 1)+1) \\
    \Gamma, \Delta & \vdash Q : n_2
  \end{align*}

  \noindent of which the first was concluded by \nameref{ts2_type_termination_loop}.
  From the premises of this rule, we get that
  \begin{align*}
    \Gamma, \Delta & \vdash e : \TINT^u_\ell \\
    \Gamma, \Delta & \vdash S : n_1'
  \end{align*}

  As we know that $v \geq 1$, this implies that $u \geq 1$, since $v \leq u$ by Lemma~\ref{lemma:safety_expressions}.
  Thus 
  \begin{equation*}
    \max(1, u (n_1' + 1)+1) = u (n_1' + 1)+1 
  \end{equation*} 

  \noindent and therefore 
  \begin{equation*} 
    n = u (n_1' + 1)+1 + n_2 
  \end{equation*}

  \noindent We can rewrite this as 
  \begin{equation*} 
    n = (u-1) (n_1' + 1) + (n_1' + 1) +1 + n_2 
  \end{equation*}

  Assume w.l.o.g.\@ that $u = v$, since this is the worst case w.r.t.\@ the number of steps required.
  Then $u-1 = v-1 = v'$, and we can therefore rewrite the above as 
  \begin{equation*} 
    n = v'(n_1' + 1) + (n_1' + 1)+1 + n_2 
  \end{equation*}

  By case \eqref{extract6} of Definition~\ref{def:extraction_function}, we have that $\EXTRACT{\code{for $e$ do $S$} :: Q} = \Delta$.
  Furthermore, $\Gamma, \Delta \vdash v' : \TINT^{v'}_{v'}$ by rule \nameref{ts2_type_termination_val}.
  We can then conclude the following:
  \begin{align*}
    \Gamma, \Delta & \vdash \code{for $v'$ do $S$} : \max(1, v'(n_1' + 1)+1)                        & \text{by \nameref{ts2_type_termination_loop}} \\
    \Gamma, \Delta & \vdash \code{for $v'$ do $S$} :: Q : \max(1, v'(n_1' + 1)+1) + n_2             & \text{by \nameref{ts2_type_termination_stm}}  \\
    \Gamma, \Delta & \vdash S :: \code{for $v'$ do $S$} :: Q : n_1' + \max(1, v'(n_1' + 1)+1) + n_2 & \text{by \nameref{ts2_type_termination_stm}}
  \end{align*}

  Say that $n' = n_1' + \max(1, v'(n_1' + 1)+1) + n_2$.
  As we know that $v$ is positive, we also know that $v'$ cannot be negative.
  We then distinguish two cases:
  \begin{enumerate}
    \item If $v' = 0$,
      then $\max(1, v'(n_1' + 1)+1) = 1$, and so 
      $n' = n_1' + 1 + n_2$.
      Clearly, $n' < n-1$; as we know that $n < g$, we therefore have that $n' < g-1$.

    \item If $v' > 0$,
      then $\max(1, v'(n_1' + 1)+1) = v' (n_1' + 1)+1$.
      Assuming still the worst case $u = v$, we thus have that
      \begin{align*}
        n  & = v'(n_1' + 1) + (n_1' + 1) + 1 + n_2 \\
        n' & = v'(n_1' + 1) + (n_1' + 1) + n_2 
      \end{align*}

      \noindent so $n' = n-1$.
      Then, as we know that $n < g$, we therefore also have that $n' < g-1$.
  \end{enumerate}

  Thus we can conclude the following:
  \begin{itemize}
    \item $\Gamma, \Delta \vdash S :: \code{for $v'$ do $S$} :: Q : n'$ as argued above.
    \item $\Gamma \vdash \ENV{S}$ by assumption, since the transition does not alter $\ENV{S}$.
    \item $\Gamma, \Delta \vdash \ENV{V}$ by assumption, since the transition does not alter $\ENV{V}$.
    \item $n' < g - 1$ as argued above.
  \end{itemize}
  \vskip1em

  %%% for e do S, false
  \item Suppose that \nameref{ts2_sss_forfalse} was used; then \eqref{eq:red} is
  \begin{equation*}
    \ENV{T}  \vdash \CONF{\code{for $e$ do $S$} :: Q, \ENV{SV}, g} \trans \CONF{Q, \ENV{SV}, g-1}
  \end{equation*}

  \noindent with $\ENV{SV} \vdash e \trans v$ and $v < 1$ as premises.
  Then $\Gamma, \Delta \vdash \code{for $e$ do $S$} :: Q : n$ was concluded by \nameref{ts2_type_termination_stm}, where $n = \max(1, u (n_1' + 1)+1) + n_2$, and with premises
  \begin{align*}
    \Gamma, \Delta & \vdash \code{for $e$ do $S$} : \max(1, u (n_1' + 1)+1) \\
    \Gamma, \Delta & \vdash Q : n_2
  \end{align*}

  \noindent of which the first was concluded by \nameref{ts2_type_termination_loop}.
  From the premises of this rule, we get that
  \begin{align*}
    \Gamma, \Delta & \vdash e : \TINT^u_\ell \\
    \Gamma, \Delta & \vdash S : n_1'
  \end{align*}

  We distinguish three cases:
  \begin{enumerate}
    \item $e$ is just a single number $v$
      (this would e.g.\@ be the case if the loop had previously executed at least once).
      In that case, $\Gamma, \Delta \vdash v : \TINT^v_v$ by rule \nameref{ts2_type_termination_val} and, as we know $v < 1$, then $\max(1, v (n_1' + 1)+1) = 1$, hence $n = 1 + n_2$;
      since $n < g$, we therefore have that $n_2 < g-1$.

    \item $e$ is a complex expression (containing operations and/or variables)
      and $u \geq 1$ (this is possible even though $v < 1$).
      Then $\max(1, u (n_1' + 1)+1) = u (n_1' + 1)+1$, hence $n = u (n_1' + 1) +1 + n_2$.
      Clearly, $1 < u (n_1' + 1) + 1$.
      As we know that $n < g$, we therefore have that $n_2 < g-1$.

    \item $e$ is a complex expression (as in the second case), but $u < 1$.
      In this case, $\max(1, u (n_1' + 1) + 1) = 1$, hence $n = 1 + n_2$ as in the first case.
      Thus, as we know that $n < g$, we therefore have that $n_2 < g-1$.
  \end{enumerate}

  By case \eqref{extract6} of Definition~\ref{def:extraction_function}, we have that $\EXTRACT{\code{for $e$ do $S$} :: Q} = \Delta$.
  Thus we can conclude the following:
  \begin{itemize}
    \item $\Gamma, \Delta \vdash Q : n_2$ by assumption.
    \item $\Gamma \vdash \ENV{S}$ by assumption, since the transition does not alter $\ENV{S}$.
    \item $\Gamma, \Delta \vdash \ENV{V}$ by assumption, since the transition does not alter $\ENV{V}$.
    \item $n_2 < g - 1$, as argued above.
  \end{itemize}
  \vskip1em

  %%% B x := e
  \item Suppose that \nameref{ts2_sss_decv} was used; then \eqref{eq:red} is
  \begin{align*}
    \ENV{T} & \vdash \CONF{\code{B $x$ := $e$ in $S$} :: Q, \ENV{SV}, g}       \\
            & \qquad \trans \CONF{S :: \DEL{x} :: Q, \ENV{S}, ((x,v,B), \ENV{V}), g-1} 
  \end{align*}

  \noindent Furthermore, $\Gamma, \Delta \vdash \code{var $x$ := $e$ in $S$} :: Q : n$ was concluded by \nameref{ts2_type_termination_stm}, where $n = n_1 + n_2$, and with premises
  \begin{align*}
    \Gamma, \Delta & \vdash \code{B $x$ := $e$ in $S$} : n_1 \\
    \Gamma, \Delta & \vdash Q : n_2
  \end{align*}

  \noindent of which the first was concluded by \nameref{ts2_type_termination_decv}.
  From the premises of this rule, we get that
  \begin{align*}
    \Gamma, \Delta        & \vdash e : B           \\ 
    \Gamma, (\Delta, x : B) & \vdash S : n_1'
  \end{align*}

  \noindent and from the conclusion that $n_1 = n_1' + 2$.
  Thus $n = n_1' + 2 + n_2$.

  From Definition~\ref{def:extraction_function}, we get that 
  \begin{align*}
    \EXTRACT{\code{B $x$ := $e$ in $S$} :: Q} & = \Delta, x:B & \text{by case \eqref{extract3}} \\
    \EXTRACT[\Delta, x:B]{\DEL{x} :: Q}       & = \Delta      & \text{by case \eqref{extract4}}
  \end{align*}

  We can then conclude the following:
  \begin{align*}
    \Gamma, (\Delta, x:B) & \vdash \DEL{x} :: Q : n_2              & \qquad\text{by \nameref{ts2_type_termination_del}} \\
    \Gamma, (\Delta, x:B) & \vdash S :: \DEL{x} :: Q : n_1' + n_2  & \qquad\text{by \nameref{ts2_type_termination_stm}} \\
    \Gamma, (\Delta, x:B) & \vdash (x, v, B), \ENV{V}              & \qquad\text{by \nameref{ts2_type_termination_envv}}
  \end{align*}

  \noindent where, in the last case, the premise $\Gamma, (\Delta, x:B) \vdash v:B$ holds by Lemma~\ref{lemma:safety_expressions}, and is concluded by \nameref{ts2_type_termination_val} (or with a subtyping rule).
  In sum, we thus have the following:
  \begin{itemize}
    \item $\Gamma, (\Delta, x:B) \vdash S :: \DEL{x} :: Q : n_1 + n_2$ as argued above, 
    \item $\Gamma \vdash \ENV{S}$ by assumption, since the transition does not alter $\ENV{S}$,
    \item $\Gamma, (\Delta, x:B) \vdash (x,v,B), \ENV{V}$ as argued above,
    \item $n_1' + n_2 < g-1$, since we know that $n_1' + 2 + n_2 < g$.
  \end{itemize}
  \vskip1em

  %%% x := e
  \item Suppose that \nameref{ts2_sss_assv} was used; then \eqref{eq:red} is 
  \begin{equation*}
    \ENV{T} \vdash \CONF{\code{$x$ := $e$} :: Q, \ENV{SV}, g} \trans \CONF{Q, \ENV{S}, \ENV{V}\EXTEND{x}{v}, g-1} 
  \end{equation*}
Furthermore, $\Gamma, \Delta \vdash \code{$x$ := $e$} :: Q : n$ was concluded by \nameref{ts2_type_termination_stm}, where $n = 1 + n_2$, and with premises
  \begin{align*}
    \Gamma, \Delta & \vdash \code{$x$ := $e$} : 1 \\
    \Gamma, \Delta & \vdash Q : n_2
  \end{align*}
 
  \noindent of which the first was concluded by \nameref{ts2_type_termination_assv}.
  From the premises of this rule, we get that $\Gamma, \Delta \vdash x : B$ and $\Gamma, \Delta \vdash e : B$.
  We know that $\ENV{SV} \vdash e \trans v$.
  We can then conclude:
  \begin{align*}
    \EXTRACT{\code{$x$ := $e$} :: Q} = & \Delta                & \qquad\text{by case \eqref{extract6} of Def.\@ \ref{def:extraction_function}} \\
    \Gamma, \Delta \vdash              & v : B                 & \qquad\text{by Lemma~\ref{lemma:safety_expressions}} \\
    \Gamma, \Delta \vdash              & \ENV{V}\EXTEND{x}{v}  & \qquad\text{by Lemma~\ref{lemma:update_variables}}
  \end{align*}

  In sum, we thus have the following:
  \begin{itemize}
    \item $\Gamma, \Delta \vdash Q : n_2$ as argued above,
    \item $\Gamma \vdash \ENV{S}$ by assumption, since the transition does not alter $\ENV{S}$,
    \item $\Gamma, \Delta \vdash \ENV{V}\EXTEND{x}{v}$ as argued above,
    \item $n_2 < g-1$, since we know that $1 + n_2 < g$.
  \end{itemize}
  \vskip1em

  %%% this.p := e
  \item Suppose that \nameref{ts2_sss_assf} was used; then \eqref{eq:red} is
  \begin{equation*}
    \ENV{T} \vdash \CONF{\code{this.$p$ := e} :: Q, \ENV{SV}, g} \trans \CONF{Q, \ENV{S}\EXTEND{X}{\ENV{P}\EXTEND{p}{v}}, \ENV{V}, g-1}  
  \end{equation*}
Furthermore, $\Gamma, \Delta \vdash \code{this.$p$ := e} :: Q : n$ was concluded by \nameref{ts2_type_termination_stm}, where $n = 1 + n_2$, and with premises
  \begin{align*}
    \Gamma, \Delta & \vdash \code{this.$p$ := e} : 1 \\
    \Gamma, \Delta & \vdash Q : n_2
  \end{align*}

  \noindent of which the first was concluded by \nameref{ts2_type_termination_assf}.
  From the premises of this rule, we get that $\Gamma, \Delta \vdash \code{this.$p$} : B$ and $\Gamma, \Delta \vdash e : B$.

  We know from the premises of \nameref{ts2_sss_assf} that $\ENV{SV} \vdash e \trans v$ and $\ENV{V}(\code{this}) = X$ and $\ENV{S}(X) = \ENV{P}$.
  We can then conclude:
  \begin{align*}
    \EXTRACT{\code{this.$p$ := $e$} :: Q} = & \Delta                                  & \qquad\text{by case \eqref{extract6} of Def.\@ \ref{def:extraction_function}} \\
    \Gamma, \Delta \vdash                   & v : B                                   & \qquad\text{by Lemma~\ref{lemma:safety_expressions}}                          \\
    \Gamma, \Delta \vdash                   & \ENV{S}\EXTEND{X}{\ENV{P}\EXTEND{p}{v}} & \qquad\text{by Lemma~\ref{lemma:update_fields}}
  \end{align*}

  In sum, we thus have the following:
  \begin{itemize}
    \item $\Gamma, \Delta \vdash Q : n_2$ as argued above,
    \item $\Gamma \vdash \ENV{S}\EXTEND{X}{\ENV{P}\EXTEND{p}{v}}$ as argued above,
    \item $\Gamma, \Delta \vdash \ENV{V}$ by assumption, since the transition does not alter $\ENV{V}$,
    \item $n_2 < g-1$, since we know that $1 + n_2 < g$.
  \end{itemize}
  \vskip1em

  %%% call e1.f(e) : e2
  \item Suppose that \nameref{ts2_sss_call} was used; then \eqref{eq:red} is
  \begin{equation*}
    \ENV{T} \vdash \CONF{\code{$e_1$.$f(\VEC{e})$:$e_2$} :: Q, \ENV{SV}, g} \trans \CONF{S :: \ENV{V} :: Q, \ENV{SV}', g-1}
  \end{equation*}
  Furthermore, $\Gamma, \Delta \vdash \code{$e_1$.$f(\VEC{e})$:$e_2$} :: Q : n$ was concluded by \nameref{ts2_type_termination_stm}, where $n = n_1 + 2 + n_2$, and with premises
  \begin{align*}
    \Gamma, \Delta & \vdash \code{$e_1$.$f(\VEC{e})$:$e_2$} : n_1 + 2 \\
    \Gamma, \Delta & \vdash Q : n_2
  \end{align*}

  \noindent of which the first was concluded by \nameref{ts2_type_termination_call}.
  From the premises of this rule, we get that
  \begin{align*}
    \Gamma, \Delta & \vdash e_1 : I            \\
    \Gamma, \Delta & \vdash \VEC{e} : \VEC{B}  \\
    \Gamma, \Delta & \vdash e_2 : \TINT^u_\ell
  \end{align*}

  \noindent and $\Gamma(I)(f) = \TPROC{\VEC{B}}^u_\ell : n_1$.

  We know from the premise of \nameref{ts2_sss_call} that $\ENV{SV} \vdash e_1 \trans X$, and by Lemma~\ref{lemma:safety_expressions} we get that $\Gamma, \Delta \vdash X : I$.

  We also know by assumption that $\Gamma \vdash \ENV{T}$, which was concluded by \nameref{ts2_type_termination_envt}.
  From its premises and side condition, we get that $\Gamma, \code{this} : I \vdash \ENV{M}$, where $\ENV{T}(X) = \ENV{M}$ by the premise of \nameref{ts2_sss_call}.
  This, in turn, was concluded by \nameref{ts2_type_termination_envm}; from the premise of this rule, we get that
  \begin{equation*}
    \Gamma, (\code{this}:I, \VEC{x}:\VEC{B}, \code{value}:\TINT^u_\ell, \code{sender}:\ITOP) \vdash S : n_1
  \end{equation*}

  \noindent and, from the premise of \nameref{ts2_sss_call}, that $\ENV{M}(f) = (f, (\VEC{x}, S))$.

  Finally, by case~\eqref{extract5} of Definition~\ref{def:extraction_function}, we have that
  \begin{equation*}
    \EXTRACT{\code{$e_1$.$f$($\VEC{e}$):$e_2$} :: Q} = \code{this}:I, \code{sender}:\ITOP, \code{value}:\TINT^u_\ell, \VEC{x}:\VEC{B} 
  \end{equation*}

  \noindent which gives us the contents of the new type environment after the transition.
  We can now conclude
  \begin{align*}
    \Gamma, \EXTRACT{\code{$e_1$.$f$($\VEC{e}$):$e_2$} :: Q} & \vdash \ENV{V} :: Q : n_2             & \qquad\text{by \nameref{ts2_type_termination_ctx}} \\
    \Gamma, \EXTRACT{\code{$e_1$.$f$($\VEC{e}$):$e_2$} :: Q} & \vdash S :: \ENV{V} :: Q : n_1 + n_2  & \qquad\text{by \nameref{ts2_type_termination_stm}} 
  \end{align*}

  In sum, we thus have the following:
  \begin{itemize}
    \item $\Gamma, \EXTRACT{\code{$e_1$.$f$($\VEC{e}$):$e_2$} :: Q} \vdash S :: \ENV{V} :: Q : n_1 + n_2$ as argued above,
    \item $\Gamma \vdash \ENV{S}$ by assumption, since the transition does not alter $\ENV{S}$,
    \item $\Gamma, \EXTRACT{\code{$e_1$.$f$($\VEC{e}$):$e_2$} :: Q} \vdash \ENV{V}'$ by simple inspection of the new $\ENV{V}'$, which contains only the bindings for the formal parameters of $f$,
    \item $n_1 + n_2 < g-1$, since we know that $n_1 + 2 + n_2 < g$.
  \end{itemize}
  \vskip1em

  %%% throw
  \item Suppose that \nameref{ts2_sss_throw} was used; then \eqref{eq:red} is 
  \begin{equation*}
    \ENV{T} \vdash \CONF{\code{throw} :: Q, \ENV{SV}, g} \trans \CONF{\EXC{\code{pge}} :: Q, \ENV{SV}, g}
  \end{equation*}

  \noindent and no further transition from the resulting configuration is then possible.
  We know that $\Gamma, \Delta \vdash \code{throw} :: Q : n$ was concluded by \nameref{ts2_type_termination_stm}, where $n = 1 + n_2$, and with premises
  \begin{align*}
    \Gamma, \Delta & \vdash \code{throw} : 1 \\
    \Gamma, \Delta & \vdash Q : n_2
  \end{align*}

  \noindent of which the first was concluded by \nameref{ts2_type_termination_throw}.
  By case~\eqref{extract6} of Definition~\ref{def:extraction_function}, we have that $\EXTRACT{\code{throw} :: Q} = \Delta$.
  Thus we can conclude the following:
  \begin{itemize}
    \item $\Gamma, \Delta \vdash \EXC{\code{pge}} :: Q : 0$ by rule \nameref{ts2_type_termination_exc},
    \item $\Gamma \vdash \ENV{S}$ by assumption, since the transition does not alter $\ENV{S}$,
    \item $\Gamma, \Delta \vdash \ENV{V}$ by assumption, since the transition does not alter $\ENV{V}$,
    \item $0 < g$, by the side condition of \nameref{ts2_sss_throw}.
  \end{itemize}
  \vskip1em

  %%% del(x)
  \item Suppose that \nameref{ts2_sss_delv} was used; then \eqref{eq:red} is 
  \begin{equation*}
    \ENV{T} \vdash \CONF{\DEL{x} :: Q, \ENV{S}, (x,v,B) : \ENV{V}, g} \trans \CONF{Q, \ENV{SV}, g}
  \end{equation*}
Furthermore, $\Gamma, \Delta, x:B \vdash \DEL{x} :: Q : n$ was concluded by \nameref{ts2_type_termination_del}.
  From its premise, we have that
  \begin{equation*}
    \Gamma, \EXTRACT[\Delta, x:B]{\DEL{x} :: Q} \vdash Q : n
  \end{equation*}

  \noindent and by case~\eqref{extract4} of Definition~\ref{def:extraction_function}, we have that
  \begin{equation*}
    \EXTRACT[\Delta, x:B]{\DEL{x} :: Q} = \Delta
  \end{equation*}

  \noindent so we know that $\Gamma, \Delta \vdash Q : n$.
  Finally, by Lemma~\ref{lemma:envv_strengthening} and $\Gamma, (\Delta, x:B) \vdash (x,v,B) : \ENV{V}$, we obtain that $\Gamma, \Delta \vdash \ENV{V}$.
  Thus we can conclude the following:
  \begin{itemize}
    \item $\Gamma, \Delta \vdash Q : n$ as argued above,
    \item $\Gamma \vdash \ENV{S}$ by assumption, since the transition does not alter $\ENV{S}$,
    \item $\Gamma, \Delta \vdash \ENV{V}$ as argued above,
    \item $n < g$ by assumption, from the side condition of \nameref{ts2_type_termination_cfg}, which was used to type the initial configuration.
  \end{itemize}
  \vskip1em

  %%% return
  \item Suppose that \nameref{ts2_sss_return} was used; then \eqref{eq:red} is 
  \begin{equation*}
    \ENV{T} \vdash \CONF{\ENV{V}' :: Q, \ENV{SV}, g} \trans \CONF{Q, \ENV{S}, \ENV{V}', g}
  \end{equation*}
  \noindent and $\Gamma, \Delta \vdash \ENV{V}' :: Q : n$ was concluded by \nameref{ts2_type_termination_ctx}.
  From its premise, we have that
  \begin{equation*}
    \Gamma, \EXTRACT{\ENV{V}' :: Q} \vdash Q : n
  \end{equation*}

  \noindent and, by cases~\eqref{extract1} and~\eqref{extract2} of Definition~\ref{def:extraction_function}, we have that $\EXTRACT{\ENV{V}' :: Q}$ yields a type environment for local variables which is formed by the variable names and the associated type information stored in $\ENV{V}'$.
  Clearly, $\Gamma, \EXTRACT{\ENV{V}' :: Q} \vdash \ENV{V}'$ for any $Q$.
  Thus we can conclude the following:
  \begin{itemize}
    \item $\Gamma, \EXTRACT{\ENV{V}' :: Q} \vdash Q : n$ as argued above,
    \item $\Gamma \vdash \ENV{S}$ by assumption, since the transition does not alter $\ENV{S}$,
    \item $\Gamma, \EXTRACT{\ENV{V}' :: Q} \vdash \ENV{V}'$ as argued above,
    \item $n < g$ by assumption, from the side condition of \nameref{ts2_type_termination_cfg}, which was used to type the initial configuration.
  \end{itemize}
\end{itemize}

In all of the above cases, we have shown that we can satisfy the premises of \nameref{ts2_type_termination_cfg} after the transition step, whenever the step is concluded by any of the ordinary transition rules for statements.
What remains now is to consider the exception rules:
\begin{itemize}
  \item Suppose the transition was concluded by \nameref{ts2_sss_oog}.
  Then \eqref{eq:red} is 
  \begin{equation*}
    \ENV{T} \vdash \CONF{S :: Q, \ENV{SV}, 0} \trans \CONF{\EXC{\code{oog}} :: S :: Q, \ENV{SV}, 0}
  \end{equation*}

  This case holds vacuously, since $g = 0$, but from the side condition of \nameref{ts2_type_termination_cfg} we know that $n < g$, which would imply that $n$ is negative, which is impossible.
  Thus $\CONF{S :: Q, \ENV{SV}, 0}$ cannot be well-typed for any $\Gamma, \Delta$.

  \item Suppose any other exception rule was used.
  Then \eqref{eq:red} is
  \begin{equation*}
    \ENV{T} \vdash \CONF{S :: Q, \ENV{SV}, g} \trans \CONF{\EXC{l} :: S :: Q, \ENV{SV}, g}
  \end{equation*}

  \noindent where $l$ is one of the two labels \code{rte} or \code{neg}.
  In either case, we have that 
  \begin{equation*}
    \Gamma, \Delta \vdash \EXC{l} :: S :: Q : 0
  \end{equation*}

  \noindent for any $\Gamma, \Delta$ by rule \nameref{ts2_type_termination_exc}.
  Thus we can conclude the following:
  \begin{itemize}
    \item $\Gamma, \Delta \vdash \EXC{l} :: S :: Q : 0$ as argued above,
    \item $\Gamma \vdash \ENV{S}$ by assumption, since the transition does not alter $\ENV{S}$,
    \item $\Gamma, \Delta \vdash \ENV{V}$ by assumption, since the transition does not alter $\ENV{V}$,
    \item $0 < g$, since we know that $g$ is positive.
  \end{itemize}
\end{itemize}

This concludes the proof.
\end{proof}

\end{document}